\documentclass[11pt]{amsart}

\usepackage{graphicx}
\usepackage{amsmath, amsthm, amssymb}
\usepackage{epstopdf}
\usepackage{epsf}
\usepackage{xcolor}
\usepackage{fancyhdr}

\newtheorem{defn}{Definition}[section]

\newtheorem{cor}[defn]{Corollary}
\newtheorem{thm}[defn]{Theorem}
\newtheorem{prop}[defn]{Proposition}
\newtheorem{lemma}[defn]{Lemma}

\newtheorem{remark}[defn]{Remark}
\newtheorem{example}[defn]{Example}

\newcommand{\be}{\begin{equation}}
\newcommand{\ee}{\end{equation}}
\newcommand{\bea}{\begin{eqnarray}}
\newcommand{\eea}{\end{eqnarray}}
\newcommand{\beas}{\begin{eqnarray*}}
\newcommand{\eeas}{\end{eqnarray*}}

\newcommand{\noi}{\noindent}

\newcommand{\goto}{\rightarrow}

\newcommand{\bp}{\begin{proof}}
\newcommand{\ep}{\end{proof}}

\title{Multidimensional Lambert-Euler inversion and vector-multiplicative coalescent processes}
\author{Yevgeniy Kovchegov}
\address{Department of Mathematics, Oregon State University, Corvallis, OR  97331, USA}
\email{kovchegy@math.oregonstate.edu}

\author{Peter T. Otto}
\address{Department of Mathematics, Willamette University, Salem, OR 97302, USA}
\email{potto@willamette.edu}

\date{}                                           

\begin{document}
\maketitle

\begin{abstract}
In this paper we show the existence of the minimal solution to the multidimensional Lambert-Euler inversion,
a multidimensional generalization of $[-e^{-1} ,0)$ branch of Lambert W function $W_0(x)$.
Specifically, for a given nonnegative irreducible symmetric matrix $V \in \mathbb{R}^{k \times k}$, 
we show that for ${\bf u}\in(0,\infty)^k$, if equation  
$$y_j  \exp\big\{\!-\!{\bf e}_j^{\sf T} V {\bf y} \big\} = u_j \qquad \forall j=1,\hdots,k,$$
has at least one solution, it must have a minimal solution ${\bf y}^*$, where the minimum is achieved in all coordinates $y_j$ simultaneously.
Moreover, such ${\bf y}^*$ is the unique solution satisfying $\rho\left(V D[y^*_j] \right) \leq 1$, 
where $D[y^*_j]={\sf diag}(y_j^*)$ is the diagonal matrix with entries $y^*_j$ and $\rho$ denotes the spectral radius.
 
Our main application is in the vector-multiplicative coalescent process.
It is a coalescent process with $k$ types of particles and vector-valued weights that begins with 
$\alpha_1n+\hdots+\alpha_k n$ particles partitioned into types of respective sizes,
and in which two clusters of weights ${\bf x}$ and ${\bf y}$ would merge with rate $({\bf x}^{\sf T} V {\bf y})/n$. 
We use combinatorics to solve the corresponding modified Smoluchowski equations, obtained as a hydrodynamic limit of vector-multiplicative coalescent
as $n \to \infty$,
and use multidimensional Lambert-Euler inversion to establish gelation and find a closed form expression for the gelation time. 

We also find the asymptotic length of the minimal spanning tree for a broad range of graphs equipped with random edge lengths.

\end{abstract}

\section{Introduction}\label{sec:intro}

In his 1783 work \cite{Euler1783} L. Euler considered the following {\it transcendental equation} entailed from 1758 work of J. H. Lambert \cite{Lambert1758}
\begin{equation}\label{eqn:EulerLambert}
x^\alpha -x^\beta=(\alpha-\beta)v x^{\alpha +\beta}.
\end{equation}
Letting $\alpha \to \beta$ in \eqref{eqn:EulerLambert}, Euler obtained
\begin{equation}\label{eqn:xLnEuler}
\ln x=v x^\beta.
\end{equation}
Next, Euler set $y=x^\beta$ and $u=\alpha v$ in \eqref{eqn:xLnEuler}, obtaining
\begin{equation}\label{eqn:yLnEuler}
{\ln y \over y}=u.
\end{equation}
Letting $y=e^w$, equation \eqref{eqn:yLnEuler} yields
\begin{equation}\label{eqn:wEuler}
we^{-w}=u.
\end{equation}
Equation \eqref{eqn:wEuler} gave rise to the Lambert W function, and in particular the function $W_0(x)$ for $-e^{-1} \leq x <0$.

\medskip
\noindent
Denote $R_0=(0,1)$, $\overline{R}_0=(0,1]$, and $R_1=(1,\infty)$. Then, for each $0<u <e^{-1}$ there are exactly two solutions $w$ of \eqref{eqn:wEuler}.
Moreover, one solution is always in $R_0$ and one solution is always in $R_1$. For $u=e^{-1}$, $w=1$ is the only solution.
Thus, for $0<u \leq e^{-1}$, there exists exactly one solution $w$ of \eqref{eqn:wEuler} in $\overline{R}_0$. 
This solution is either unique when $u=1$ or is the smaller of the two solutions when $0<u <e^{-1}$.

\medskip
\noindent
Lambert-Euler inversion \eqref{eqn:wEuler} yields the existence of function 
\begin{equation}\label{defx}
x(t):=\min\{x>0~:~xe^{-x}=te^{-t}\}, \qquad t\in (0,\infty),
\end{equation}
with the range $\overline{R}_0$. 
In 1960, function $x(t)$ was used by P. Erd\H{o}s and A. R\'enyi \cite{ER60} for establishing formation of a giant cluster in the theory of random graphs.
In 1962, J. B. McLeod \cite{McLeod62} used Lambert-Euler inversion and function $x(t)$ in the analysis of {\it Smoluchowski coagulation equations} 
with multiplicative kernel (aka {\it Flory coagulation system}), observing the {\it gelation phenomenon}.

\medskip
\noindent
In this paper, we will study the multidimensional Lambert-Euler inversion problem.
Let $V \in \mathbb{R}^{k \times k}$ be a nonnegative irreducible symmetric matrix.
For a given vector ${\bf z}\in(0,\infty)^k$, consider region
\begin{equation}\label{eqn:R0}
R_0 =\left\{{\bf z}  \in (0,\infty)^k :\, \rho\left(V D[z_j]\right) <1 \right\}
\end{equation}
its closure within $(0,\infty)^k$,
\begin{equation}\label{eqn:R0bar}
\overline{R}_0=\left\{{\bf z}  \in (0,\infty)^k :\, \rho\left(V D[z_j] \right) \leq 1\right\},
\end{equation}
and the complement of $\overline{R}_0$ within $(0,\infty)^k$,
\begin{equation}\label{eqn:R1}
R_1 =\left\{{\bf z}  \in (0,\infty)^k :\, \rho\left(V D[z_j] \right) >1 \right\},
\end{equation}
where for a vector ${\bf x} \in \mathbb{R}^k$ with coordinates $x_i$,
$D[x_i]$ denotes the diagonal matrix with entries $x_i$, and $\rho(M)$ denotes the spectral radius of matrix $M$.

\medskip 
\noindent
In this paper we found it convenient to us bra-ket notation of P. Dirac. Specifically, 
$|{\bf x}\rangle$ will denote the column vector representation of vector ${\bf x} \in \mathbb{R}^k$,
and $\langle {\bf x}|$ will denote the row vector representation of vector ${\bf x} \in \mathbb{R}^k$.
For $c \in \mathbb{R}$ and ${\bf x} \in \mathbb{R}^k$, $\,c|{\bf x}\rangle$ will represent the product $c{\bf x}$, a column vector.
Respectively, $\langle {\bf x}|{\bf y}\rangle=\langle {\bf y}|{\bf x}\rangle$ will be the dot product of ${\bf x}$ and ${\bf y}$ in $\mathbb{R}^k$.
Finally, for a matrix $M \in \mathbb{R}^{k \times k}$, $\langle {\bf x}|M|{\bf y}\rangle$ will represent the product ${\bf x}^{\sf T}M{\bf y}$ resulting in a scalar.

\medskip 
\noindent
Let ${\bf e}_j$ denote the $j$-th standard basis vector. The following theorem is the main result of the paper. 
\begin{thm}[Multidimensional Lambert-Euler inversion]\label{thm:multEulerLambert}
Consider a nonnegative irreducible symmetric matrix $V \in \mathbb{R}^{k \times k}$.
For any given ${\bf z}\in(0,\infty)^k$, there exists a unique vector ${\bf y}\in \overline{R}_0$ such that
\begin{equation}\label{eqn:yzSolsCoord}
y_j  e^{-\langle {\bf e}_j  | V | {\bf y} \rangle} = z_j  e^{-\langle {\bf e}_j  | V | {\bf z} \rangle} \qquad j=1,\hdots,k.
\end{equation}
Moreover, if ${\bf z}\in \overline{R}_0$, then ${\bf y}={\bf z}$. If ${\bf z}\in R_1$, 
then ${\bf y}<{\bf z}$ ($y_i <z_i$ $\forall i$), i.e., ${\bf y}$ is the smallest solution of \eqref{eqn:yzSolsCoord}.
\end{thm}

\medskip 
\noindent
Theorem \ref{thm:multEulerLambert} which we will prove in Section \ref{sec:euler}, 
yields the following multidimensional analogue of $[-e^{-1} ,0)$ branch of Lambert $W$ function $W_0(x)$.
Consider domain 
\be\label{eqn:DR}
\mathcal{D}=\left\{{\bf u} \in (0,\infty)^k ~~:~~\exists \,{\bf z} \in (0,\infty)^k  \quad\text{ such that }\quad
|{\bf u}\rangle =\sum\limits_{j=1}^k z_j \,  e^{-\langle {\bf e}_j  | V | {\bf z} \rangle}|{\bf e}_j\rangle  \right\}.
\ee
Then, by Theorem \ref{thm:multEulerLambert},
\be\label{eqn:DR0}
\mathcal{D}=\left\{{\bf u} \in (0,\infty)^k ~~:~~\exists \,{\bf z} \in  \overline{R}_0  \quad\text{ such that }\quad
|{\bf u}\rangle =\sum\limits_{j=1}^k z_j \,  e^{-\langle {\bf e}_j  | V | {\bf z} \rangle}|{\bf e}_j\rangle  \right\}
\ee
and for the mapping $\,\Psi_V:(0,\infty)^k \to \mathcal{D}\,$ defined as $\,|\Psi_V(z) \rangle=\sum\limits_{j=1}^k z_j \,  e^{-\langle {\bf e}_j  | V | {\bf z} \rangle}|{\bf e}_j\rangle\,$, 
Theorem \ref{thm:multEulerLambert}  implies that restricting the domain of $\,\Psi_V\,$ to $\overline{R}_0$, makes $\,\Psi_V\,$ a continuous bijection from $\overline{R}_0$ to $\mathcal{D}$. Thus, we can define Lambert-Euler inversion $\,\Lambda_V\,$ of $\,\Psi_V\,$ as a continuous bijection from $\mathcal{D}$ to $\overline{R}_0$. 
Thus, for any given ${\bf z}\in(0,\infty)^k$, equation \eqref{eqn:yzSolsCoord}
has the minimal solution ${\bf y}=\Lambda_V \circ \Psi_V ({\bf z}) \in \overline{R}_0$, where the minimum is achieved in all coordinates.

\medskip 
\noindent
Now, we can define a multidimensional analogue of function $x(t)$ in \eqref{defx}.
For any given ${\boldsymbol \alpha}\in(0,\infty)^k$, let  ${\bf y}(t)=\Lambda_V \circ \Psi_V({\boldsymbol \alpha} t)$ for all $t>0$, 
i.e., ${\bf y}(t)$ is the minimal solution of 
\be\label{eqn:GenMin}
y_j  e^{-\langle {\bf e}_j  | V | {\bf y} \rangle} = \alpha_j  t e^{-t \langle {\bf e}_j  | V | {\boldsymbol \alpha} \rangle} \qquad \forall j=1,\hdots,k.
\ee
Consequently, $\,{\bf y}(t)=\Lambda_V \circ \Psi_V({\boldsymbol \alpha} t)$ is a continuous function.

\medskip
\noindent
By analogy with  $\,\lim\limits_{t \to \infty}{x(t) \over t}=0$ for  $x(t)$ in \eqref{defx}, we will show  
that $\,\lim\limits_{t \to \infty}{{\bf y}(t) \over t}={\boldsymbol 0}$. This will be done in Lemma \ref{lem:massdissipation} of Section \ref{sec:euler}.

\medskip 
\noindent
Next, we will list the applications of Theorem \ref{thm:multEulerLambert}.

\medskip

\subsection{Vector-Multiplicative Coalescent Processes}\label{intro:coalescent} 
The solution to the multidimensional Lambert-Euler inversion given in Theorem \ref{thm:multEulerLambert} and
the function ${\bf y}(t)$ defined in \eqref{eqn:GenMin} will be used in the analysis of a general class of coalescent processes introduced here 
that we will call the {\it  vector-multiplicative coalescent processes}.

\medskip
\noindent 
For a nonnegative irreducible symmetric matrix $V \in \mathbb{R}^{k \times k}$ and a given vector ${\boldsymbol \alpha}\in(0,\infty)^k$, 
let 
\be\label{eqn:aplhan}
{\boldsymbol \alpha}[n]=\left[\!\!\begin{array}{c}\alpha_1[n] \\ \vdots \\ \alpha_k[n]\end{array}\!\!\right]={\boldsymbol \alpha}n+o(\sqrt{n}).
\ee
Consider a system with $k$ types of particles, $1,\hdots,k$,
and the coalescent process that begins with $\langle {\boldsymbol \alpha}[n]|{\bf 1}\rangle=\langle {\boldsymbol \alpha}|{\bf 1}\rangle n+o(\sqrt{n})$ singletons
distributed between the $k$ types so that for each $i$, there are $\alpha_i[n]=\alpha_i n+o(\sqrt{n})$ particles of type $i$.
In this continuous time Markov process, a particle of type $i$ bonds with a particle of type $j$ with the rate $v_{i,j}/n$, 
where $v_{i,j}=\langle {\bf e}_i|V|{\bf e}_j\rangle$ is the $(i,j)$ element in the matrix $V$.
The bonds are formed independently. This process is called {\it vector-multiplicative coalescent}.

\medskip
\noindent
Formally, vector-multiplicative coalescent process describes cluster merger dynamics, where the weight of each cluster is a $k$-dimensional vector ${\bf x} \in \mathbb{Z}_+^k$ such that $\langle {\bf x}|{\bf 1}\rangle>0$. 
Each cluster of weight ${\bf x}$ bonds together $x_1,\hdots,x_k$ particles of corresponding types $1,\hdots,k$.  
The coalescent process begins with $\langle {\boldsymbol \alpha}[n]|{\bf 1}\rangle$ singletons of all $k$ types, of which there are $\alpha_i[n]$ of type $i$ (for all $i=1,\hdots,k$). Each pair of clusters with respective weight vectors ${\bf x}$ and ${\bf y}$ would coalesce into a cluster of weight ${\bf x+y}$ with rate $K({\bf x},{\bf y})/n$,
where 
\be\label{VMkernel} 
K({\bf x}, {\bf y}) = \langle {\bf x} | V | {\bf y} \rangle.  
\ee
The last merger will create a cluster of weight $\,{\boldsymbol \alpha}[n]$.

\medskip
\noindent
The kernel $K({\bf x}, {\bf y})$ defined in \eqref{VMkernel} will be referred to as the {\it vector-multiplicative kernel}.
The kernel is symmetric
$$K({\bf x}, {\bf y}) =K({\bf y},{\bf x}) \quad \text{ for all vectors }~{\bf x}, {\bf y}$$
and bilinear
$$K(c_1{\bf x} + c_2{\bf y}, {\bf z}) = c_1\,K({\bf x}, {\bf z}) + c_2\,K({\bf y}, {\bf z}) \quad \text{ for all vectors }~{\bf x}, {\bf y}, {\bf z} ~~\text{ and scalars }~c_1,c_2.$$

\medskip
\noindent
Notice that coalescent processes with vector-valued weights have been considered in the past. See \cite{KOY,KBN98,VZ98}.

\medskip
\noindent
In the vector-multiplicative coalescent process, let $\zeta_{\bf x}^{[n]}(t)$ denote the number of clusters of weight ${\bf x}$ at time $t \geq 0$. 
The initial values are $\zeta_{\bf x}^{[n]}(0)=\sum\limits_{i=1}^k \alpha_i[n] \delta_{{\bf e}_i,{\bf x}}$.
The process $${\bf ML}_n(t)=\Big(\zeta_{\bf x}^{[n]}(t) \Big)_{{\bf x} \in \mathbb{Z}_+^k : \langle {\bf x}|{\bf 1}\rangle>0}$$
that counts clusters of all types in the vector-multiplicative coalescent process is the corresponding   
{\it Marcus-Lushnikov process}. 
In Lemma \ref{lem:hydroFlory} we will refer to the known weak limit result of T.\,G. Kurtz for {\it density dependent population processes} that yields 
convergence of $\zeta_{\bf x}^{[n]}(t)$ to $\zeta_{\bf x}(t)$, where $\zeta_{\bf x}(t)$ is the solution to the {\it modified Smoluchowski equations} (MSE)
$${d \over dt}\zeta_{\bf x}(t) =  - \zeta_{\bf x} \langle {\bf x} | V | \boldsymbol{\alpha} \rangle  + \frac{1}{2}\sum_{{\bf y}, {\bf z} \,: {\bf y} + {\bf z} = {\bf x}} \langle {\bf y} | V | {\bf z} \rangle \zeta_{{\bf y}} \zeta_{{\bf z}}$$
with the same initial conditions $\zeta_{\bf x}(0)=\sum\limits_{i=1}^k \alpha_i \delta_{{\bf e}_i,{\bf x}}$.

\medskip
\noindent
In Section \ref{sec:mse} we will find the unique solution $\zeta_{\bf x}(t)$ of the above modified Smoluchowski equations.
Specifically, for a vector ${\bf x} \in \mathbb{Z}_+^k$ let
$\,{\bf x}!=x_1!x_2!\hdots x_k!\,$
and for vectors ${\bf a}$ and ${\bf b}$ in $\mathbb{R}^k$ let
$\,{\bf a}^{\bf b}=a_1^{b_1} a_2^{b_2}\hdots a_k^{b_k}\,$ whenever $a_i^{b_i}$ is uniquely defined for all $i$.
Now, consider a complete graph $K_k$ consisting of vertices $\{1,\hdots,k\}$ with weights $w_{i,j}=w_{j,i}\geq 0$ assigned to
its edges $[i,j]$ ($i \not=j$). 
Let the weight $W(\mathcal{T})$ of a spanning tree $\mathcal{T}$ be the product of the weights of all of its edges.
Finally, let $\tau(K_k, w_{i,j})=\sum\limits_{\mathcal{T}}W(\mathcal{T})$ denote the {\it weighted spanning tree enumerator}, 
i.e., the sum of weights of all spanning trees in $K_k$.
These notations are used in the following closed form expression, that will be established in Corollary \ref{cor:TxsolODE} of Section \ref{sec:mse}
\be\label{eqn:MSEsol}
\zeta_{\bf x}(t) = {1 \over {\bf x}!}\boldsymbol{\alpha}^{\bf x} {\tau(K_k, x_i x_j v_{i,j}) \over {\bf x}^{\bf 1}}(V {\bf x})^{{\bf x}- {\bf 1}} e^{-\langle {\bf x} | V | \boldsymbol{\alpha} \rangle t} t^{\langle {\bf x} |\boldsymbol{1} \rangle -1}.
\ee

\bigskip
The concept of {\it gelation} was studied in \cite{Aldous98,Aldous,Jeon98,Jeon99,Rezakhanlou13,Spouge83,vDE86,Ziff} and related papers.
For the hydrodynamic limit $\zeta_{\bf x}(t)$ of the Marcus-Lushnikov process with vector-multiplicative kernel,
the {\it gelation time} $T_{gel}$ is the time after which the total mass $\,\sum_{\bf x} \zeta_{\bf x}(t) |{\bf x} \rangle \,$ begins to dissipate, i.e.,
while the initial total mass was $\,\sum_{\bf x} \zeta_{\bf x}(0) |{\bf x} \rangle= |\boldsymbol{\alpha} \rangle$,
 $$T_{gel}=\inf\Big\{t>0\,:\,\sum_{\bf x} \zeta_{\bf x}(t) |{\bf x} \rangle  <  |\boldsymbol{\alpha} \rangle \Big\}.$$
In Section \ref{sec:gel} we will use Corollary \ref{cor:ytat} of the Multidimensional Lambert-Euler inversion Theorem \ref{thm:multEulerLambert}
for establishing {\it gelation} in the coagulation process with the vector-multiplicative kernel \eqref{VMkernel} and find the value of the {\it gelation time}.
Specifically, we will show that the gelation time equals to
$$T_{gel}={1 \over \rho ( V D[\alpha_i] )}.$$

\bigskip

\subsection{Application in minimal spanning trees}\label{intro:mst}
Another application concerns the problem of finding the asymptotic mean length of the minimal spanning tree in a graph equipped with independent random edge lengths
as studied in \cite{BFMcD,CFIJS,Frieze,Janson95} and related research works.
Let $K_{{\boldsymbol \alpha}[n]}$ denote a graph with $\langle {\boldsymbol \alpha}[n]|{\bf 1}\rangle=\langle {\boldsymbol \alpha}|{\bf 1}\rangle+o(\sqrt{n})$
vertices divided into $k$ partitions of respective sizes 
$$\alpha_1[n],\hdots,\alpha_k[n],$$ 
where, each vertex in the $i$-th partition is connected with each vertex in the $j$-th partition by an edge if and only if $v_{i,j}=v_{j,i}>0$.
Even within an $i$-th partition, if $v_{i,i}>0$, a pair of vertices is connected by an edge.

\medskip
\noindent
Now, we equip the edges in the graph $K_{{\boldsymbol \alpha}[n]}$ with edge lengths as follows. 
For each edge $e$ connecting a vertex in the $i$-th partition with a vertex in the $j$-th partition
we have an associated random variable $\ell_e \sim {\sf Beta}\left(1,v_{i,j}\right)$, distributed on $(0,1)$ via 
the beta probability density function 
$$f_{i,j}(x)=v_{i,j}(1-x)^{v_{i,j}-1}, \qquad 0<x<1.$$
Random variables $\{\ell_e\}_e$ are sampled independently. 
Here, $\ell_e$ represents the length of edge $e$. 

\medskip
\noindent
The length of a tree graph is the sum of the lengths $\ell_e$ of its edges.
Consider the minimal spanning tree of $K_{{\boldsymbol \alpha}[n]}$, i.e., a spanning tree of $K_{{\boldsymbol \alpha}[n]}$ with the minimal length.
Let random variable $L_n$ denote the length of the minimal spanning tree of $K_{{\boldsymbol \alpha}[n]}$. 
We are interested in finding the limit $\,\lim_{n \goto \infty} \mathbb{E}[L_n]$.
Following the approach in \cite{KOY}, $\,\lim_{n \goto \infty} \mathbb{E}[L_n]$ will be expressed using the closed form solution 
\eqref{eqn:MSEsol} of the modified Smoluchowski equations.

\medskip
\noindent
The connection between the limit $\,\lim_{n \goto \infty} \mathbb{E}[L_n]\,$ and $\zeta_{\bf x}(t)$ is established using the following framework.
First, we construct a random graph process by considering a ``time" parameter $p \in [0,1]$, and declaring an edge $e$ ``open" if $\ell_e \leq p$
and ``closed'' if $\ell_e > p$. 
Thus, at time $p$, we have a graph consisting of $\langle {\boldsymbol \alpha}[n]|{\bf 1}\rangle$ vertices and all open edges.
The random graph process $G(n,p)$ describes the corresponding percolation dynamics on graph $K_{{\boldsymbol \alpha}[n]}$ equipped with edge lengths $\ell_e$. 
Process $G(n,p)$ partitions $K_{{\boldsymbol \alpha}[n]}$ into clusters of vertices connected by open edges at time $p$.

\medskip
\noindent
Notice that under the time change $p=1-e^{-t/n}$ with $t \in [0,\infty)$, for an edge $e$ connecting a vertex in the $i$-th partition with a vertex in the $j$-th partition we have
$$P(\ell_e \leq 1-e^{-t/n})=1-e^{-tv_{i,j}/n}.$$
Thus, in $G(n,1-e^{-t/n})$, an edge adjacent to a vertex in the $i$-th partition and a vertex in the $j$-th partition would open after waiting 
for an exponentially distributed arrival time with parameter $v_{i,j}/n$.

\medskip
\noindent
There is a one-to-one correspondence between connected clusters in $G(n,p)$ and vectors in $\mathbb{Z}_+^k$. 
Namely, a cluster with $x_i$ vertices in the $i$-th partition is represented by a vector ${\bf x}$ in $\mathbb{Z}_+^k$ with coordinates $x_i$.
Consider the random graph process $G(n,1-e^{-t/n})$.
For two clusters composed of two disjoint subsets of vertices in $K_{{\boldsymbol \alpha}[n]}$, represented by vectors ${\bf x}$ and ${\bf y}$, 
the waiting time for the clusters to connect via an open edge will be  
an exponential random variable with parameter $\langle {\bf x}|V|{\bf y}\rangle/n$.
Hence, the merger dynamics of clusters in the random graph process $G(n,1-e^{-t/n})$ matches the merger dynamics of clusters in 
the vector-multiplicative coalescent process. That is, if we let $\xi_{\bf x}^{[n]}(t)$ denote the number of clusters represented by vector ${\bf x}$
in $G(n,1-e^{-t/n})$, then $\,\Big(\xi_{\bf x}^{[n]}(t) \Big)_{\bf x}$ is distributed as the Marcus-Lushnikov process
$\,{\bf ML}_n(t)=\Big(\zeta_{\bf x}^{[n]}(t) \Big)_{\bf x}$.


\medskip
\noindent
As $\zeta_{\bf x}^{[n]}(t)$, and therefore, $\xi_{\bf x}^{[n]}(t)$ converges weakly to $\zeta_{\bf x}(t)$, the approach in \cite{KOY} yields 
$$\lim_{n \goto \infty} \mathbb{E}[L_n] = \sum_{\!{\bf x} : \langle {\bf x}|{\bf 1}\rangle>0~} \int\limits_0^\infty \zeta_{\bf x}(t) \, dt.$$

\medskip
\noindent
In Corollary \ref{cor:LnKxV}, the closed form expression \eqref{eqn:MSEsol} of the solution $\zeta_{\bf x}(t)$ to the modified Smoluchowski equations
is substituted, yielding the following general asymptotic equation
$$\lim_{n \goto \infty} \mathbb{E}[L_n] = \sum_{{\bf x}: \langle {\bf x}|{\bf 1}\rangle>0} {(\langle {\bf x} |\boldsymbol{1} \rangle -1)! \over {\bf x}!}\boldsymbol{\alpha}^{\bf x} {\tau(K_k, x_i x_j v_{i,j}) \over {\bf x}^{\bf 1}}(V {\bf x})^{{\bf x}- {\bf 1}} \, \langle {\bf x} | V | \boldsymbol{\alpha} \rangle^{-\langle {\bf x} | {\bf 1} \rangle}.$$

\bigskip
\noindent
Additionally, recalling a known correspondence between the gelation time $T_{gel}$ in the Marcus-Lushnikov process and the time $p_c$ of formation
of a giant component in $G(n,p)$, we have $p_c \sim 1-e^{-T_{gel}/n} \sim {T_{gel} \over n}$. Hence,
$$p_c \sim {1 \over n\,\rho (V D[\alpha_i] )}.$$

\bigskip
\noindent
Finally, random graph processes $G(n,1-e^{-t/n})$ have many features similar to the {\it inhomogeneous random graphs} formalism considered in
S\"{o}derberg \cite{Soderberg},    Bollob\'{a}s et all \cite{BJR2007}, and related papers.
We hope that the results of this current manuscript can be used in the study of inhomogeneous random graphs

\medskip
\noindent
We conclude the introduction by noticing that the parts of the paper on Smoluchowski coagulation equations and on spanning trees 
are tied to a number of interesting Abel's type multinomial identities such as \eqref{eqn:abeltau}. 

\bigskip

\section{Multidimensional Lambert-Euler inversion}\label{sec:euler}
In this section we will prove Theorem \ref{thm:multEulerLambert} that enables the multidimensional Lambert-Euler inversion.
Additionally, in Lemma \ref{lem:massdissipation} we will show that for ${\boldsymbol \alpha}\in(0,\infty)^k$, $\,\lim\limits_{t \to \infty}{{\bf y}(t) \over t}={\boldsymbol 0}$.

\medskip
\noindent
For ${\bf z}\in(0,\infty)^k$, let function $F({\bf x};{\bf z})$ be defined as follows 
\begin{equation*}
| F({\bf x};{\bf z}) \rangle=| {\bf x} \rangle -\sum\limits_{j=1}^k z_j  e^{\langle {\bf e}_j  | V | {\bf x}-{\bf z} \rangle} |{\bf e}_j\rangle, \qquad \forall {\bf x} \in (0,\infty)^k.
\end{equation*}
Notice that a root ${\bf x}$ of $F({\bf x};{\bf z})$ is a solution of \eqref{eqn:yzSolsCoord}.
Next, for a given ${\bf z}\in(0,\infty)^k$, we find the Jacobian matrix of $F({\bf x};{\bf z})$ in the equation below
\begin{align}\label{eqn:fieldFjacobian}
{\partial F({\bf x};{\bf z}) \over \partial {\bf x}}&=I-\sum\limits_{i,j=1}^k z_j  e^{\langle {\bf e}_j  | V | {\bf x}-{\bf z} \rangle} |{\bf e}_j\rangle \langle {\bf e}_j  | V | {\bf e}_i \rangle \langle {\bf e}_i  |
=I-\sum\limits_{j=1}^k z_j  e^{\langle {\bf e}_j  | V | {\bf x}-{\bf z} \rangle} |{\bf e}_j\rangle \langle {\bf e}_j  | V  \nonumber \\
&=I-D\left[z_j  e^{\langle {\bf e}_j  | V | {\bf x}-{\bf z} \rangle}\right]V.
\end{align}

\medskip
\noindent
We will need the following trivial proposition.
\begin{prop}\label{prop:rhoBigger}
For a given pair of vectors $\,{\bf a} \not= {\bf b}$ in $(0,\infty)^k$, if $a_i \leq b_i$ for all coordinates $i$, then
$\rho(V D[a_j]) < \rho(V D[b_j]).$
\end{prop}

\medskip
\noindent
The following lemma is instrumental for establishing uniqueness in Theorem \ref{thm:multEulerLambert}.
\begin{lemma}\label{lem:R0}
If ${\bf y},{\bf z} \in \overline{R}_0$ satisfy equation \eqref{eqn:yzSolsCoord}, then ${\bf y}={\bf z}$.
\end{lemma}
\begin{proof}
We will prove this statement by contradiction. Suppose there is a pair ${\bf y},{\bf z} \in \overline{R}_0$ 
satisfying ${\bf y}\not={\bf z}$ and  \eqref{eqn:yzSolsCoord}. Then, $F({\bf y};{\bf z})=0$.
Consider a set of indices $\,\mathcal{I}=\{i\,:\, y_i<z_i\}$, and let $\mathcal{I}^c=\{i\,:\, y_i\geq z_i\}$ denote its complement.

\medskip
\noindent
Next, consider a point ${\bf x}^*$ with coordinates
\begin{equation*}
|{\bf x}^*\rangle = |{\bf z}\rangle+\sum\limits_{i \in \mathcal{I}}(y_i-z_i)|{\bf e}_i\rangle
\end{equation*}
and a path ${\bf x}(t),$ $\,t \in [0,1]$, connecting ${\bf z}$ and ${\bf y}$ defined as
\begin{equation}\label{eqn:path_xtCases}
|{\bf x}(t)\rangle = \begin{cases}
      |{\bf z}\rangle +2t|{\bf x}^*-{\bf z}\rangle & \text{ for } t \in [0,1/2], \\
      |{\bf x}^*\rangle +(2t-1)|{\bf y}-{\bf x}^*\rangle & \text{ for } t \in [1/2,1].
\end{cases}
\end{equation}

\medskip
\noindent
Now, since $F({\bf y};{\bf z})=F({\bf z};{\bf z})=0$, then from \eqref{eqn:fieldFjacobian} and \eqref{eqn:path_xtCases} we have
\begin{align}\label{eqn:int}
0 =& F({\bf y};{\bf z})-F({\bf z};{\bf z}) =\int\limits_0^1 \big | dF({\bf x}(t);{\bf z}) \big\rangle =\int\limits_0^1 {\partial F({\bf x};{\bf z}) \over \partial {\bf x}} ~\Big| \, d{\bf x}(t) \Big\rangle \nonumber \\
=&\int\limits_0^{1/2} {\partial F({\bf x};{\bf z}) \over \partial {\bf x}} ~\Big| \, d{\bf x}(t) \Big\rangle
+\int\limits_{1/2}^1 {\partial F({\bf x};{\bf y}) \over \partial {\bf x}} ~\Big| \, d{\bf x}(t) \Big\rangle \nonumber \\
=& 2\!\!\int\limits_0^{1/2} \Big( I-D\left[z_j  e^{2t\langle {\bf e}_j  | V | {\bf x}^*-{\bf z} \rangle}\right]V \Big)  \, \Big| \, {\bf x}^*-{\bf z} \Big\rangle \,dt
+2\!\!\int\limits_{1/2}^1 \Big( I-D\left[y_j  e^{2(t-1)\langle {\bf e}_j  | V |  {\bf y} -{\bf x}^* \rangle}\right]V  \Big)\, \Big| \, {\bf y}-{\bf x}^* \Big\rangle \,dt  \nonumber \\
&\qquad = |{\bf y}-{\bf z}\rangle -|\mathcal{E}_I\rangle -|\mathcal{E}_{II}\rangle,
\end{align}
where 
\begin{equation}\label{eqn:E1}
| \mathcal{E}_I \rangle=2\!\!\int\limits_0^{1/2} D\left[z_j  e^{2t\langle {\bf e}_j  | V | {\bf x}^*-{\bf z} \rangle}\right]V  \, \Big| \, {\bf x}^*-{\bf z} \Big\rangle \,dt
=D[z_j a_j]V | {\bf x}^*-{\bf z} \rangle
\end{equation}
with
$$a_j=\begin{cases}
      {1-e^{-\langle {\bf e}_j  | V | {\bf z}-{\bf x}^* \rangle} \over \langle {\bf e}_j  | V | {\bf z}-{\bf x}^* \rangle} & \text{ if } j \in \mathcal{I}, \\
      1 & \text{ if } j \in \mathcal{I}^c,
\end{cases}$$
and similarly,
\begin{equation}\label{eqn:E2}
| \mathcal{E}_{II}\rangle=2\!\!\int\limits_{1/2}^1  D\left[y_j  e^{2(t-1)\langle {\bf e}_j  | V |  {\bf y} -{\bf x}^* \rangle}\right]V  \, \Big| \, {\bf y}-{\bf x}^* \Big\rangle \,dt
= D[y_j b_j]V  | {\bf y}-{\bf x}^* \rangle
\end{equation}
with
$$b_j=\begin{cases}
      1 & \text{ if } j \in \mathcal{I}, \\
      {1-e^{-\langle {\bf e}_j  | V | {\bf y}-{\bf x}^* \rangle} \over \langle {\bf e}_j  | V | {\bf y}-{\bf x}^* \rangle} & \text{ if } j \in \mathcal{I}^c.
\end{cases}$$

\medskip
\noindent
Notice that for $j \in \mathcal{I}$, we have $\langle {\bf e}_j  | V | {\bf z}-{\bf x}^* \rangle >0$ and  $\, a_j={1-e^{-\langle {\bf e}_j  | V | {\bf z}-{\bf x}^* \rangle} \over \langle {\bf e}_j  | V | {\bf z}-{\bf x}^* \rangle} \in (0,1)$.
Similarly, for $j \in \mathcal{I}^c$, we have $\langle {\bf e}_j  | V | {\bf y}-{\bf x}^* \rangle>0$ and $\, b_j={1-e^{-\langle {\bf e}_j  | V | {\bf y}-{\bf x}^* \rangle} \over \langle {\bf e}_j  | V | {\bf y}-{\bf x}^* \rangle} \in (0,1)$.
Thus, by Proposition~\ref{prop:rhoBigger},
\begin{equation}\label{eqn:rhoBigger}
\rho\left(D[z_j  a_j]V\right) \leq \rho\left(D[z_j]V\right)
\quad
\text{ and }
\quad
\rho\left(D[y_j  b_j]V\right)  \leq \rho\left(D[y_j]V\right),
\end{equation}
where the first inequality is strict if ${\bf z}\not={\bf x}^*$ and the second inequality is strict if ${\bf y}\not={\bf x}^*$.

\medskip
\noindent 
  Now, since $|\mathcal{E}_I \rangle$ in \eqref{eqn:E1} has all nonpositive coordinates
  and  $| \mathcal{E}_{II} \rangle$ in \eqref{eqn:E2} has all nonnegative coordinates, we have
  $$\langle \mathcal{E}_I \,|\, \mathcal{E}_{II} \rangle =\langle \mathcal{E}_{II} \,|\, \mathcal{E}_I \rangle \leq 0.$$
  Therefore, since $\rho\left(D[y_j]V\right) \leq 1$ and either ${\bf y}\not={\bf x}^*$ or ${\bf z}\not={\bf x}^*$ or both hold, 
  equations \eqref{eqn:E1}, \eqref{eqn:E2}, and \eqref{eqn:rhoBigger} imply
  \begin{align*}
  \big\|\mathcal{E}_I+\mathcal{E}_{II}\big\|^2 &\leq \big\|\mathcal{E}_I\big\|^2 +\big\|\mathcal{E}_{II}\big\|^2 
  ~< \rho^2\!\left(D[z_j]V\right) \big\|{\bf z}-{\bf x}^*\big\|^2 +\rho^2\!\left(D[y_j]V\right)\big\| {\bf y}-{\bf x}^*\big\|^2 \\
  &\leq  \big\|{\bf z}-{\bf x}^*\big\|^2 +\big\| {\bf y}-{\bf x}^*\big\|^2 ~=\big\| {\bf y}-{\bf z}\big\|^2
  \end{align*}
  as $\,\langle {\bf y}-{\bf x}^* \, | \, {\bf z}-{\bf x}^* \rangle=0$. 
  The contradiction to equation \eqref{eqn:int} follows.
\end{proof}

\medskip
\noindent
Let $|{\bf 1}\rangle=\sum\limits_{i=1}^k | {\bf e}_i \rangle$ denote the vector with all of its coordinates equal to $1$, and
let $|{\bf 0}\rangle$ denote the vector of zeros. 
For a vector ${\bf x} \in (0,\infty)^k$ with coordinates $x_i$, let $| {\bf x}^{-1} \rangle=\sum\limits_{i=1}^k x_i^{-1} | {\bf e}_i \rangle$ denote the vector with coordinates $x_i^{-1}$. Also, for vectors ${\bf a}$ and ${\bf b}$ in $\mathbb{R}^k$, we will write ${\bf a}<{\bf b}$ if $a_i < b_i$ for all $i$. 
Similarly, for matrices $A$ and $B$ in $\mathbb{R}^{k \times k}$, we will write $A<B$ if the inequality holds coordinate-wise.

\medskip
\noindent
We will need the following lemma. 
\begin{lemma}\label{lem:R1toR0}
For any given ${\bf z}\in R_1$, there exists a vector ${\boldsymbol \eta}=\sum\limits_{i=1}^k \eta_i | {\bf e}_i \rangle$ such that 
$$0<\eta_i<1 \qquad \forall i,$$
$\rho\left(V D[z_j \eta_j] \right)=1$, and
\begin{equation}\label{eqn:eta1}
V D[z_j]| {\bf 1}-{\boldsymbol \eta} \rangle=| {\boldsymbol \eta}^{-1} \rangle -| {\bf 1}\rangle.
\end{equation}
\end{lemma}
\begin{proof}
Let ${\bf u}>{\boldsymbol 0}$ (i.e., $u_i>0$ $\forall i$) be the Perron-Frobenius eigenvector of $V D[z_j]$, then since $\rho\left(V D[z_j] \right)>1$,
we have
$$\langle {\bf e}_i | V D[z_j] -I | {\bf u}\rangle ~>0 \quad \text{ for all }~i,$$
and therefore
\begin{equation}\label{eqn:u1}
\langle {\bf e}_i |  V D[z_j]| {\bf u}  \rangle ~>u_i \quad \text{ for all }~i.
\end{equation}
Consider two sequences of vectors in $(0,\infty)^k$, ${\boldsymbol \eta}^{(k)}$ and ${\bf w}^{(k)}$ evolving according to the following recursion 
\begin{equation}\label{eqn:wetaRec}
\eta_i^{(k)}={1 \over 1+ \langle {\bf e}_i |  V D[z_j]| {\bf w}^{(k-1)}  \rangle} \quad \text{ and }\quad {\bf w}^{(k)}={\bf 1}-{\boldsymbol \eta}^{(k)}.
\end{equation}
Let the sequences begin with ${\bf w}^{(0)}=1-{\boldsymbol \eta}^{(0)}=\varepsilon {\bf u}$ with $\varepsilon>0$ sufficiently small so that
${\bf w}^{(0)} <{\bf 1}$, and by \eqref{eqn:u1},
$$\eta_i^{(1)}={1 \over 1+ \langle {\bf e}_i |  V D[z_j]| {\bf w}^{(0)}  \rangle}=1- \langle {\bf e}_i |  V D[z_j]| {\bf w}^{(0)} \rangle+O(\varepsilon^2) 
~<~ 1-w^{(0)}_i=\eta_i^{(0)}  \qquad \text{ for all }~i.$$
Then, ${\boldsymbol \eta}^{(0)}>{\boldsymbol \eta}^{(1)}$, and by \eqref{eqn:wetaRec}, ${\bf w}^{(1)}>{\bf w}^{(0)}$, which in turn implies
${\boldsymbol \eta}^{(1)}>{\boldsymbol \eta}^{(2)}$, and so on. Recursively obtaining
$${\boldsymbol \eta}^{(k)}>{\boldsymbol \eta}^{(k+1)} \quad \text{ and }\quad {\bf w}^{(k+1)}>{\bf w}^{(k)}$$
for all $k=0,1,\hdots$. Hence, the limits 
$$\lim\limits_{k \to \infty} {\boldsymbol \eta}^{(k)}= {\boldsymbol \eta}  \quad \text{ and }\quad \lim\limits_{k \to \infty} {\bf w}^{(k)}={\bf w}$$
satisfy 
\begin{equation}\label{eqn:weta}
{\boldsymbol 0}<{\boldsymbol \eta}={\bf 1}-{\bf w}< {\bf 1} \qquad \text{ and }\qquad \eta_i={1 \over 1+ \langle {\bf e}_i |  V D[z_j]| {\bf w} \rangle}  \quad \text{ for all }~i.
\end{equation}
Equation \eqref{eqn:weta} implies
$$\eta_i \langle {\bf e}_i |  V D[z_j]| {\bf w} \rangle=1-\eta_i=w_i  \quad \text{ for all }~i.$$

\medskip 
\noindent
Thus, $\eta_i w_i^{-1}=\langle {\bf e}_i |  V D[z_j]| {\bf w} \rangle^{-1}=\langle {\bf e}_i |  V D[z_j w_j]| {\bf 1} \rangle^{-1}\,$, and
\begin{equation}\label{eqn:etaVDz}
D[\eta_j w_j^{-1}]VD[z_j w_j] | {\bf 1} \rangle=| {\bf 1} \rangle ,
\end{equation}
i.e., all the rows of $D[\eta_j w_j^{-1}]VD[z_j w_j]$ add up to $1$.

\medskip 
\noindent
For the vector $| {\bf w}{\boldsymbol \eta}^{-1} \rangle=\sum\limits_{i=1}^k w_i \eta_i^{-1} | {\bf e}_i \rangle$,
equation \eqref{eqn:etaVDz} yields
$$V D[z_j \eta_j]| {\bf w}{\boldsymbol \eta}^{-1} \rangle=D[w_j \eta_j^{-1}] D[\eta_j w_j^{-1}]VD[z_j w_j] | {\bf 1} \rangle=D[w_j \eta_j^{-1}]| {\bf 1} \rangle=| {\bf w}{\boldsymbol \eta}^{-1} \rangle.$$
Hence, by Perron-Frobenius theorem, $1$ is the Perron-Frobenius eigenvalue of $V D[z_j \eta_j]$, the spectral radius $\rho\left(V D[z_j \eta_j] \right)=1$, 
and $| {\bf w}{\boldsymbol \eta}^{-1} \rangle$ is the corresponding Perron-Frobenius eigenvector.

\medskip 
\noindent
Therefore, as $\,{\bf 1}-{\boldsymbol \eta}={\bf w}$,
$$VD[z_j] | {\bf 1}-{\boldsymbol \eta} \rangle=VD[z_j \eta_j] | {\bf w}{\boldsymbol \eta}^{-1} \rangle=| {\bf w}{\boldsymbol \eta}^{-1} \rangle=| {\boldsymbol \eta}^{-1} \rangle -| {\bf 1}\rangle$$
yielding the claim in \eqref{eqn:eta1}.
\end{proof}

\bigskip
\begin{lemma}\label{lem:R1toR0}
For any given ${\bf z}\in R_1$, there exists a unique vector ${\bf y}\in R_0$ such that \eqref{eqn:yzSolsCoord} is satisfied.
Moreover, ${\bf y}<{\bf z}$.
\end{lemma}
\begin{proof}
For ${\bf x} \in (0,\infty)^k$ with coordinates $x_i$, let $| \ln {\bf x}\rangle=\sum\limits_{i=1}^k \ln x_i | {\bf e}_i \rangle$ denote the vector with coordinates $\ln x_i$. 
Consider function $f_z: (0,\infty)^k \to \mathbb{R}^k$ defined as follows:
\begin{equation}\label{eqn:fz}
| f_z({\boldsymbol \xi})\rangle= |\ln {\boldsymbol \xi} \rangle+VD[z_i] | {\bf 1}-{\boldsymbol \xi} \rangle, \qquad {\boldsymbol \xi} \in (0,\infty)^k.
\end{equation}
For the vector ${\boldsymbol \eta}$ in Lemma \ref{lem:R1toR0}, we have
\begin{equation*}\label{eqn:fzeta}
| f_z({\boldsymbol \eta})\rangle = |\ln {\boldsymbol \eta} \rangle+VD[z_i] | {\bf 1}-{\boldsymbol \eta} \rangle= -|\ln {\boldsymbol \eta}^{-1} \rangle+| {\boldsymbol \eta}^{-1} \rangle -| {\bf 1}\rangle \
=-\sum\limits_{i=1}^k \ln\left(\eta_i^{-1} e^{1-\eta_i^{-1}}\right) | {\bf e}_i \rangle \,>\,{\bf 0}
\end{equation*}
since $\eta_i^{-1}>1$ $~\forall i~$ and $xe^{1-x}<1$ for all positive $x \not=1$.

\medskip 
\noindent
Now, since $| f_z({\boldsymbol \eta})\rangle>{\boldsymbol 0}$, by continuity of $f_z({\boldsymbol \xi})$, there exists $\delta\in(0,1)$ sufficiently small so that $\widetilde{\boldsymbol \eta}=(1-\delta){\boldsymbol \eta}$ satisfies  
$$| f_z(\widetilde{\boldsymbol \eta})\rangle>{\boldsymbol 0}.$$
Notice that since $\rho\left(V D[z_j \eta_j] \right)=1$, we have $\rho\left(V D[z_j \widetilde{\eta}_j] \right)=1-\delta<1$.

\medskip 
\noindent
Next, consider a smooth curve ${\bf x}(t)$ in $(0,\infty)^k$ that begins at ${\bf x}(0)=\widetilde{\boldsymbol \eta}$ and evolves according to
the following differential equations:
\begin{equation}\label{eqn:xtODE}
{d \over dt}x_i(t)=-x_i(t) \sum\limits_{m=0}^\infty \langle {\bf e}_i | \big(V D[z_j x_j(t)] \big)^m | f_z(\widetilde{\boldsymbol \eta})\rangle \qquad i=1,\hdots,k.
\end{equation}
As each $x_i(t)$ is monotone decreasing to $0$ at exponentially fast rate as $t \to \infty$, and since $\rho\big(V D[z_j x_j(0)] \big)<1$, 
by Prop.~\ref{prop:rhoBigger}, we have $\rho\big(V D[z_j x_j(t)] \big)<1$ for all $t \geq 0$. Thus,
$$\sum\limits_{m=0}^\infty \big(V D[z_j x_j(t)] \big)^m=\big(I-V D[z_j x_j(t)] \big)^{-1}$$
is well-defined for all $t \geq 0$. Therefore, \eqref{eqn:xtODE} yields
\begin{equation*}\label{eqn:logxtODE}
{d \over dt}|\ln {\bf x}(t) \rangle=-\big(I-V D[z_j x_j(t)] \big)^{-1} | f_z(\widetilde{\boldsymbol \eta})\rangle
\end{equation*}
and
\begin{align}\label{eqn:fzteta}
\big| f_z\big({\bf x}(t)\big)\big\rangle &= | f_z(\widetilde{\boldsymbol \eta})\rangle+|\ln {\bf x}(t) \rangle - |\ln \widetilde{\boldsymbol \eta} \rangle +VD[z_i] | \widetilde{\boldsymbol \eta}-{\bf x}(t) \rangle \nonumber \\
&= | f_z(\widetilde{\boldsymbol \eta})\rangle+\int\limits_0^t \Big( D\big[1/x_j(s)\big]-VD[z_i]\Big) \big| d{\bf x}(s) \rangle \nonumber \\
&= | f_z(\widetilde{\boldsymbol \eta})\rangle+\int\limits_0^t \big(I-V D[z_j x_j(s)] \big) |d\ln {\bf x}(s) \rangle \nonumber \\
&= | f_z(\widetilde{\boldsymbol \eta})\rangle- \int\limits_0^t  | f_z(\widetilde{\boldsymbol \eta})\rangle ds ~=(1-t)| f_z(\widetilde{\boldsymbol \eta})\rangle.
\end{align}
Hence, substituting $t=1$ into \eqref{eqn:fzteta} yields
\begin{equation}\label{eqn:x10}
\big| f_z\big({\bf x}(1)\big)\big\rangle={\boldsymbol 0},
\end{equation}
where ${\bf x}(1)<{\bf x}(0)=\widetilde{\boldsymbol \eta}<{\boldsymbol \eta} <{\bf 1}.$

\medskip 
\noindent
Next, we let $|{\bf y} \rangle=D[z_j ]|{\bf x}(1)\rangle=\sum\limits_{i=1}^k z_i x_i(1) | {\bf e}_i \rangle$, then by Prop.~\ref{prop:rhoBigger},
$$\rho\big(V D[y_j] \big)=\rho\big(VD[z_j x_j(1)]\big)<\rho\left(V D[z_j \widetilde{\eta}_j] \right)<1,\quad \text{ i.e., }~{\bf y}\in R_0.$$ 
 Also,  ${\bf x}(1)<{\bf 1}$ yields
$${\bf y} <{\bf z}.$$

\medskip 
\noindent
Finally, equations \eqref{eqn:x10} and $x_i(1)=y_i z_i^{-1}$ imply
$${\boldsymbol 0}=\big| f_z\big({\bf x}(1)\big)\big\rangle
=|\ln {\bf x}(1)\rangle +VD[z_i] \big | {\bf 1}-{\bf x}(1) \big\rangle
=|\ln {\bf y}\rangle -|\ln {\bf z} \rangle+V|{\bf z}-{\bf y} \rangle$$
arriving at 
$$|\ln {\bf y}\rangle -V|{\bf y} \rangle=|\ln {\bf z} \rangle -V|{\bf z} \rangle,$$
equivalent to equation \eqref{eqn:yzSolsCoord}.
\end{proof}

\bigskip 
\noindent
Lemmas \ref{lem:R0} and \ref{lem:R1toR0} yield the following simple corollary.
\begin{cor}\label{cor:R1R0}
If $\rho\left(V D[z_j] \right)=1$, then ${\bf y}={\bf z}$ is the only solution of \eqref{eqn:yzSolsCoord}.
\end{cor}

\bigskip 
\noindent
Now, we have a proof of the Multidimensional Lambert-Euler inversion.
\begin{proof}[Proof of Theorem \ref{thm:multEulerLambert}]
The statement in Theorem \ref{thm:multEulerLambert} follows immediately from Lemmas \ref{lem:R0} and \ref{lem:R1toR0}, and Corollary \ref{cor:R1R0}.
\end{proof}

\medskip 
\noindent
Theorem \ref{thm:multEulerLambert} yields the following corollary.
\begin{cor}\label{cor:ytat} 
Let ${\bf y}(t)=\Lambda_V \circ \Psi_V({\boldsymbol \alpha} t)$ be the minimal solution of \eqref{eqn:GenMin}.  Then,
\begin{itemize}
  \item[(a)] $~~~{\bf y}(t) = {\boldsymbol \alpha} t\,$ for all $\,t \leq  \frac{1}{\rho ( V D[\alpha_i] )}$;
  \item[(b)] $~~~{\bf y}(t) < {\boldsymbol \alpha} t\,$ for all $\,t >  \frac{1}{\rho ( V D[\alpha_i] )}$.
\end{itemize}
\end{cor}

\medskip 
\noindent
Notice that $\,\lim\limits_{t \to \infty}{x(t) \over t}=0$ for  $x(t)$ in \eqref{defx} which, in the context of random graphs, 
is analogous to the absorption of connected components of various sizes 
by a giant component in Erd\H{o}s-R\'enyi random graph model \cite{ER60}. 
We have the corresponding multidimensional result.
\begin{lemma}\label{lem:massdissipation}
For any given ${\boldsymbol \alpha}\in(0,\infty)^k$, let ${\bf y}(t)=\Lambda_V \circ \Psi_V({\boldsymbol \alpha} t)$ be the minimal solution of \eqref{eqn:GenMin}. Then,
$$\lim\limits_{t \to \infty}{{\bf y}(t) \over t}={\boldsymbol 0}.$$
\end{lemma}

\medskip
\begin{proof}
Let $v_{i,j}=\langle {\bf e}_i  | V | {\bf e}_j \rangle$ denote the entries in matrix $V$. Recall that $V$ is nonnegative irreducible symmetric matrix.
Thus, $v_{i,j}=v_{j,i} \geq 0$ for all $i,j$. 

\medskip 
\noindent
Recall that $\,y_i(t) \leq \alpha_i t\,$ for all $i$ and all $t >0$.
Equation \eqref{eqn:GenMin} implies
\be\label{eqn:yovert}
{y_i(t) \over t}=\alpha_i e^{-t \langle {\bf e}_i  | V | {\boldsymbol \alpha}-{\bf y}/t \rangle} \quad \text{ for all }~i=1,\hdots,k,
\ee
where ${\bf y}={\bf y}(t)$. First, we claim that
\be\label{eqn:lsupayt}
\limsup\limits_{t \to \infty} {1 \over t} \langle {\bf 1} | {\bf y} \rangle <\langle {\bf 1} |  {\boldsymbol \alpha} \rangle.
\ee
We prove \eqref{eqn:lsupayt} by contradiction as follows. Suppose, not. Then, there exists a sequence $t_m>0$
increasing to $\infty$, such that
$$\lim\limits_{m \to \infty} {1 \over t_m} \langle {\bf 1} | {\bf y}(t_m) \rangle =\langle {\bf 1} |  {\boldsymbol \alpha} \rangle.$$
Hence, for all $m$ sufficiently large, $y_i(t_m)>\alpha_i t_m/2$ for all $i=1,\hdots,k$. 
Thus, since ${\bf y}(t)\in \overline{R}_0$, the spectral radius
$$1\geq \rho(VD[y_j]) \geq {t_m \over 2} \rho(VD[\alpha_j])$$
by Prop.~\ref{prop:rhoBigger}, contradicting $\lim\limits_{m \to \infty}t_m=\infty$.
Therefore, equation \eqref{eqn:lsupayt} holds.

\medskip 
\noindent
By equation \eqref{eqn:lsupayt}, there exists $\varepsilon>0$ so small that it satisfies
\be\label{eqn:vare1}
\varepsilon < \min\limits_j \alpha_j 
\quad \text{ and }\quad 
\limsup\limits_{t \to \infty} {1 \over t} \langle {\bf 1} | {\bf y} \rangle <\langle {\bf 1} | {\boldsymbol \alpha} \rangle -k\varepsilon.
\ee
Next, \eqref{eqn:vare1} implies the existence of $T>0$ large enough so that whenever $t>T$ we have
\be\label{eqn:vare3}
\alpha_j -\alpha_i e^{-v_{i,j}\varepsilon t} >\varepsilon \quad \text{ for all }~i,j
\ee
and
\be\label{eqn:vare2}
\forall t>T \quad \exists j' ~~\text{ such that }~~~{y_{j'}(t) \over t}<\alpha_{j'}-\varepsilon.
\ee

\medskip 
\noindent
Next, for a given $t>T$, we show that if $\,{y_j(t) \over t}<\alpha_j-\varepsilon \,$ for some $j$, then for all $i$ such that $v_{i,j}>0$,  we have
$${y_i(t) \over t}<\alpha_i-\varepsilon.$$
Indeed, equations \eqref{eqn:yovert} and \eqref{eqn:vare3} yield
$${y_i(t) \over t}=\alpha_i e^{-t \langle {\bf e}_i  | V | {\boldsymbol \alpha}-{\bf y}/t \rangle} \leq \alpha_i e^{-t v_{i,j}(\alpha_j-y_j/t)}
< \alpha_i e^{-v_{i,j}\varepsilon t}<\alpha_i-\varepsilon.$$ 

\medskip 
\noindent
Hence, by \eqref{eqn:vare2} and irreducibility of $V$, for all $t>T$, we have
\be\label{eqn:vare2all}
{y_i(t) \over t}<\alpha_i-\varepsilon \quad \text{ for all }~i.
\ee

\medskip 
\noindent
Together, equations \eqref{eqn:yovert} and \eqref{eqn:vare2all}  imply
$${y_i(t) \over t} \leq \alpha_i e^{-t v_{i,j}(\alpha_j-y_j/t)}<\alpha_i e^{-v_{i,j}\varepsilon t} \quad \text{ for all }~t>T \quad \text{ and all }~i,j.$$
Thus, by irreducibility of $V$, we have $\lim\limits_{t \to \infty}{y_i(t) \over t}=0$ exponentially fast for each $i=1,\hdots,k$.
\end{proof}

\bigskip

\section{Vector-Multiplicative Coalescent Processes} 

In this section we will analyze Smoluchowski coagulation equations \eqref{eqn:SE} and modified Smoluchowski equations \eqref{eqn:Flory}.
In Lemma \ref{lem:hydroFlory}, we will show that equations \eqref{eqn:Flory} are a hydrodynamic limit of the Marcus-Lushnikov process 
for the vector-multiplicative coalescent.
Our main result is in Subsection \ref{sec:mse}, where we will use tools from combinatorics and linear algebra to find a complete solution to the modified Smoluchowski system of equations \eqref{eqn:Flory}.

\medskip

\subsection{Vector-Multiplicative Smoluchowski Equations}\label{sec:Smol} 

\medskip
Consider a vector-multiplicative coalescent process introduced in Subsection \ref{intro:coalescent}.
Let $\zeta_{\bf x}(t)$ be an averaged quantity that tracks the relative number of clusters of weight ${\bf x}$ at tome $t\geq 0$.  
Since the process evolves according to the merger rates $\,n^{-1}\langle {\bf x} | V | {\bf y} \rangle$, 
the {\it Smoluchowski coagulation system} of equations for the vector-multiplicative coalescent process is written as follows:
\begin{equation}\label{eqn:SE}
{d \over dt}\zeta_{\bf x}(t) = - \zeta_{\bf x}  \sum_{{\bf y}} \zeta_{{\bf y}} \langle {\bf x} | V | {\bf y} \rangle + \frac{1}{2}\sum_{{\bf y}, {\bf z} \,: {\bf y} + {\bf z} = {\bf x}} \langle {\bf y} | V | {\bf z} \rangle \zeta_{{\bf y}} \zeta_{{\bf z}} 
\end{equation}
with the initial conditions $\zeta_{\bf x}(0)=\sum\limits_{i=1}^k \alpha_i \delta_{{\bf e}_i,{\bf x}}$.
Functions $\zeta_{\bf x}$ are indexed by all weight vectors ${\bf x} \in \mathbb{Z}_+^k$ satisfying $\langle {\bf x}|{\bf 1}\rangle>0$. 
This is also the domain for summation, i.e., $\sum\limits_{\bf x} f({\bf x})=\!\!\sum\limits_{{\bf x}\in \mathbb{Z}_+^k:\langle {\bf x}|{\bf 1}\rangle>0}f({\bf x})$.

\medskip
\noindent
Note that the initial conditions $\zeta_{\bf x}(0)=\sum\limits_{i=1}^k \alpha_i \delta_{{\bf e}_i,{\bf x}}$ yield $\,\sum_{\bf x} \zeta_{\bf x}(0) |{\bf x} \rangle = |\boldsymbol{\alpha} \rangle$.
Equation \eqref{eqn:SE} implies
\be\label{eqn:consmassSE}
{d \over dt}\sum_{\bf x} \zeta_{\bf x}(t) |{\bf x} \rangle =- \sum_{\bf x}\zeta_{\bf x}(t)|{\bf x} \rangle  \sum_{{\bf y}} \zeta_{{\bf y}} \langle {\bf x} | V | {\bf y} \rangle + \frac{1}{2}\sum_{{\bf y}, {\bf z}} \langle {\bf y} | V | {\bf z} \rangle \zeta_{{\bf y}} \zeta_{{\bf z}}|{\bf y} + {\bf z} \rangle=0
\ee
whenever the second order moments of the solutions $\zeta_{\bf x}(t)$ of \eqref{eqn:SE} are convergent, i.e., the matrix of all second order moments
$\,A(t)=\sum\limits_{\bf x} \zeta_{\bf x}(t) |{\bf x}\rangle \langle {\bf x}|$ has all finite entries.

\medskip
\noindent
If we set the total mass constant by letting
$$\sum_{\bf y} \zeta_{\bf y}(t) |{\bf y} \rangle = |\boldsymbol{\alpha} \rangle, \qquad \text{ where }~~\langle \boldsymbol{\alpha}| = (\alpha_1, \alpha_2, \ldots, \alpha_k),$$
in the right hand side of \eqref{eqn:SE}, then equation \eqref{eqn:SE} will turn into the following quasilinear system of equations
\be
\label{eqn:Flory}
{d \over dt}\zeta_{\bf x}(t) =  - \zeta_{\bf x} \langle {\bf x} | V | \boldsymbol{\alpha} \rangle  + \frac{1}{2}\sum_{{\bf y}, {\bf z} \,: {\bf y} + {\bf z} = {\bf x}} \langle {\bf y} | V | {\bf z} \rangle \zeta_{{\bf y}} \zeta_{{\bf z}},
\ee
with the same initial conditions $\zeta_{\bf x}(0)=\sum\limits_{i=1}^k \alpha_i \delta_{{\bf e}_i,{\bf x}}$ as in \eqref{eqn:SE}. 
Equations of the type in \eqref{eqn:Flory} are called {\it modified Smoluchowski equations} (MSE) or {\it Flory} system of equations.

\medskip
\noindent
Equation \eqref{eqn:consmassSE} implies that the solutions of Smoluchowski coagulation equations \eqref{eqn:SE} and modified Smoluchowski equations \eqref{eqn:Flory}
will coincide as long as the second order moments $\,A(t)=\sum\limits_{\bf x} \zeta_{\bf x}(t) |{\bf x}\rangle \langle {\bf x}|$ of the solutions $\zeta_{\bf x}(t)$ of \eqref{eqn:SE} are convergent, i.e., for all $t$ between $0$ and $t_c$, where
$$t_c =\inf\Big\{t>0~:~\sum\limits_{\bf x} \zeta_{\bf x}(t) |{\bf x}\rangle \langle {\bf x}| \text{ diverges } \Big\}.$$

\bigskip

\subsection{Marcus-Lushnikov process and hydrodynamic limit}
Recall that {\it Marcus-Lushnikov process} ${\bf ML}_n(t)$ keeps track of cluster counts in the vector-multiplicative coalescent process that begins 
with $\langle {\boldsymbol \alpha}[n]|{\bf 1}\rangle$ singletons of $k$ types with $\alpha_i[n]$ of type $i$ for all $i$. Specifically, let $\zeta_{\bf x}^{[n]}(t)$ denote the number of connected components of weight ${\bf x}$ at time $t$. Then, 
$${\bf ML}_n(t)=\Big(\zeta_{\bf x}^{[n]}(t) \Big)_{{\bf x} \in \mathbb{Z}_+^k : \langle {\bf x}|{\bf 1}\rangle>0}$$
with the starting values $\zeta_{\bf x}^{[n]}(0)=\sum\limits_{i=1}^k \alpha_i[n] \delta_{{\bf e}_i,{\bf x}}$.

\medskip
\noindent
Our next lemma states that the solution to the modified Smoluchowski coagulation system \eqref{eqn:Flory} is 
the hydrodynamic limit of the Marcus-Lushnikov process ${\bf ML}_n(t)$ with cross-multiplicative kernel. 
\begin{lemma}\label{lem:hydroFlory}
For any given $T>0$ and all ${\bf x} \in \mathbb{Z}_+^k$ satisfying $\langle {\bf x}|{\bf 1}\rangle>0$,
$$\lim\limits_{n \to \infty} \sup\limits_{s \in [0,T]} \left|n^{-1}\zeta_{\bf x}^{[n]}(s)-\zeta_{\bf x}(s)\right|=0 \qquad \text{ a.s.}$$
where $\zeta_{\bf x}(t)$ is the solution of the modified Smoluchowski coagulation system \eqref{eqn:Flory}
with the initial conditions $\zeta_{\bf x}(0)=\sum\limits_{i=1}^k \alpha_i \delta_{{\bf e}_i,{\bf x}}$. 
\end{lemma}
\begin{proof}
The proof is an application of the weak convergence results of T.\,G. Kurtz for {\it density dependent population processes}. 
Namely, Theorem 2.1 in Chapter 11 of \cite{EK}, or equivalently, Theorem 8.1 in \cite{Kurtz81}.
This lemma follows immediately from the approach in Section 5 of \cite{KOY} 
by replacing $V=\left[\begin{array}{cc}0 & 1 \\1 & 0\end{array}\right]$
with any other nonnegative irreducible symmetric matrix $V \in \mathbb{R}^{k \times k}$.
\end{proof}

\bigskip

\subsection{Solving the Modified Smoluchowski Equations}\label{sec:mse} 
Recall that 
$\,\boldsymbol{\alpha}^{\bf x} = \alpha_1^{x_1} \alpha_2^{x_2} \cdots \alpha_k^{x_k}$.
The following proposition generalizes the approach in \cite{KOY, McLeod62}.
\begin{prop}\label{prop:solODE}
Consider 
\be
\label{eqn:testsoln} 
\zeta_{\bf x}(t) = \boldsymbol{\alpha}^{\bf x} S_{\bf x} e^{-\langle {\bf x} | V | \boldsymbol{\alpha} \rangle t} t^{\langle {\bf x} |\boldsymbol{1} \rangle -1}. 
\ee
with $S_{\bf x}$ solving the following reccursion  
\be\label{SolnEqn} 
S_{\bf x} (\langle {\bf x} | \boldsymbol{1} \rangle -1) = \frac{1}{2}\sum_{{\bf y}, {\bf z} : {\bf y} + {\bf z} = {\bf x}} \langle {\bf y} | V | {\bf z} \rangle S_{\bf y} S_{\bf z}
\ee
with the initial conditions $S_{{\bf e}_j} = 1$ for all $j=1,\hdots,k$. 
Then, $\zeta_{\bf x}(t)$ is the unique solution of MSE \eqref{eqn:Flory} with the initial conditions $\zeta_{\bf x}(0)=\sum\limits_{i=1}^k \alpha_i \delta_{{\bf e}_i,{\bf x}}$.
\end{prop}
\begin{proof}
First we show that $\zeta_{\bf x}(t)$ is a solution of \eqref{eqn:Flory}.
Differentiating with respect to $t$ yields
\beas
{d \over dt}\zeta_{\bf x}(t) & = & \boldsymbol{\alpha}^{\bf x} S_{\bf x} \left[ e^{-\langle {\bf x} | V | \boldsymbol{\alpha} \rangle t} (\langle {\bf x} | \boldsymbol{1} \rangle -1) t^{\langle {\bf x} | \boldsymbol{1} \rangle -2} - \langle {\bf x} | V | \boldsymbol{\alpha} \rangle e^{-\langle {\bf x} | V | \boldsymbol{\alpha} \rangle t} t^{\langle {\bf x} | \boldsymbol{1} \rangle -1} \right] \\
& = & \boldsymbol{\alpha}^{\bf x} S_{\bf x} e^{-\langle {\bf x} | V | \boldsymbol{\alpha} \rangle t} (\langle {\bf x} | \boldsymbol{1} \rangle -1) t^{\langle {\bf x} | \boldsymbol{1} \rangle -2} - \langle {\bf x} | V | \boldsymbol{\alpha} \rangle \zeta_{\bf x}(t) ,
\eeas
and for ${\bf y} + {\bf z} = {\bf x}$,
\beas
\zeta_{{\bf y}} (t) \zeta_{{\bf z}} (t) & = & \boldsymbol{\alpha}^{{\bf y}} \boldsymbol{\alpha}^{{\bf z}} S_{{\bf y}} S_{{\bf z}} e^{-(\langle {\bf y} | V | \boldsymbol{\alpha} \rangle + \langle {\bf z} | V | \boldsymbol{\alpha} \rangle) t} t^{\langle {\bf y} | \boldsymbol{1} \rangle + \langle {\bf z} | \boldsymbol{1} \rangle - 2} \\
& = & \boldsymbol{\alpha}^{\bf x} S_{{\bf y}} S_{{\bf z}} e^{-\langle {\bf x} | V | \boldsymbol{\alpha} \rangle t} t^{\langle {\bf x} | \boldsymbol{1} \rangle - 2}. 
\eeas

\noindent
Plugging the above two equations into \eqref{eqn:Flory} yields \eqref{SolnEqn}.

\medskip
\noindent
Finally, the uniqueness of solution \eqref{eqn:testsoln} of \eqref{eqn:Flory} follows from quasilinearity of \eqref{eqn:Flory}. 
\end{proof}

\medskip
\noi
Next, we complete the solution of \eqref{eqn:Flory} by finding a combinatorial expression for $S_{\bf x}$ in \eqref{eqn:testsoln}.
First, we need the following notations.

\medskip
\noindent
For a given ${\bf x} \in \mathbb{Z}_+^k$ satisfying $\langle {\bf x}|{\bf 1}\rangle>0$, let $K_{\bf x}(V)$ denote a graph equipped with
edge weights such that
\begin{itemize}
  \item $\,K_{\bf x}(V)$  is a complete graph with  $\langle {\bf x}|{\bf 1}\rangle$ vertices; its vertices are 
  partitioned into $k$ groups with the number of vertices in the $i$-th partition set equal to $x_i$, the $i$-th coefficient of the vector ${\bf x}$;  
  \item $\,V$ is the matrix of edge weights, i.e., the weight of an edge connecting a vertex in the $i$-th partition set with a vertex in the $j$-th partition set equals $v_{i,j}=v_{j,i}$.
\end{itemize}
Finally, if $\mathcal{T}$ is a spanning tree of $\,K_{\bf x}(V)$, then the weight of $\mathcal{T}$ is the product of the weights of all of its edges.
Let $T_{\bf x}=T_{\bf x}(V)$ denote the {\it weighted spanning tree enumerator} of $\,K_{\bf x}(V)$, i.e., $T_{\bf x}$ is the sum of weights of all spanning trees of $\,K_{\bf x}(V)$.
Now, for a graph consisting of just one vertex, the weighted spanning tree enumerator is set to be equal $1$. Thus, $T_{{\bf e}_j}=1$ for all $j=1,\hdots,k$.
\begin{lemma}\label{lem:solSx}
Let $T_{\bf x}$ be the weighted spanning tree enumerator of $\,K_{\bf x}(V)$. Then, 
$$S_{\bf x} = \frac{T_{\bf x}}{{\bf x}!} \qquad \text{ where we denote }~~{\bf x}!=x_1!x_2!\hdots x_k!$$
is the solution to the recursion equation \eqref{SolnEqn}.
\end{lemma}  
\begin{proof}
Let us count the total weight $T_{\bf x}=T_{\bf x}(V)$ of all spanning trees of $\,K_{\bf x}(V)$. For a given ${\bf y}$ and ${\bf z}$ satisfying ${\bf y} + {\bf z} = {\bf x}$, there are
$${x_1 \choose y_1} {x_2 \choose y_2} \cdots {x_k \choose y_k}$$
ways of splitting $K_{\bf x}(V)$ into $K_{\bf y}(V)$ and $K_{\bf z}(V)$. 
Each of the two subgraphs, $K_{\bf y}(V)$ and $K_{\bf z}(V)$ has the respective weighted spanning tree enumerators $T_{\bf y}$ and $T_{\bf z}$.
For any given disection of $K_{\bf x}(V)$ into $K_{\bf y}(V)$ and $K_{\bf z}(V)$, the total weight of the edges connecting the two subgraphs equals $\langle {\bf y} | V | {\bf z} \rangle$.
Now, there are $\langle {\bf x} | \boldsymbol{1}\rangle - 1$ edges in every spanning tree of $\,K_{\bf x}(V)$, and each edge splits the tree into two spanning trees,
$K_{\bf y}(V)$ and $K_{\bf z}(V)$.
Hence, the total weight $T_{\bf x}$ of all spanning trees in $\,K_{\bf x}(V)$ satisfies
\be\label{eqn:recTx}
T_{\bf x} = \frac{1}{2 (\langle {\bf x} | \boldsymbol{1}\rangle - 1)} \sum_{{\bf y}, {\bf z} : {\bf y} + {\bf z} = {\bf x}} {x_1 \choose y_1} {x_2 \choose y_2} \cdots {x_k \choose y_k} \langle {\bf y} | V | {\bf z} \rangle T_{\bf y} T_{\bf z},
\ee
where the multiple of ${1 \over 2}$ accounts for double counting  ${\bf y}+{\bf z}$ splits with ${\bf z}+{\bf y}$ splits.

\medskip
\noindent
Equation \eqref{eqn:recTx} can be rewritten as
\be\label{eqn:TxSoln} 
{T_{\bf x} \over {\bf x}!} (\langle {\bf x} | {\bf 1} \rangle -1) = \frac{1}{2}\sum_{{\bf y}, {\bf z} : {\bf y} + {\bf z} = {\bf x}} \langle {\bf y} | V | {\bf z} \rangle {T_{\bf y} \over {\bf y}!}  {T_{\bf z} \over {\bf z}!} 
\ee
with the initial conditions $T_{{\bf e}_j}=1$ for all $j=1,\hdots,k$.
Therefore, by the uniqueness of the solution of the recursive equation \eqref{SolnEqn}, we have $\,S_{\bf x} = \frac{T_{\bf x}}{{\bf x}!}$.
\end{proof}

\bigskip
\noindent 
Let $L_{\bf x}=L_{\bf x}(V)$ denote the {\it weighted Laplacian matrix} of $\,K_{\bf x}(V)$, i.e., 
$L_{\bf x}=\big(l_{r,s}\big) \in \mathbb{R}^{\langle {\bf x} | \boldsymbol{1}\rangle\times \langle {\bf x} | \boldsymbol{1}\rangle}$ 
is a matrix with coordinates
\be\label{eqn:Lx}
l_{r,s}=\begin{cases}
      \langle {\bf e}_i |V| {\bf x} \rangle-v_{i,i} & \text{ if } r=s, \,\text{ where }\, i=1+\max\{m:\,s_m < r \}, \\
      -v_{i,j} & \text{ if } r\not=s, \,\text{ where }\, i=1+\max\{m:\,s_m < r \}, ~j=1+\max\{m:\,s_m < s\},
\end{cases}
\ee
where $s_0=0$, and $s_m=\sum\limits_{i=1}^m x_i$ for $m=1,\hdots,k$. Schematically, $L_{\bf x}$ is represented as follows

$\hskip 0.28 in \overbrace{ \hskip 1.5 in }^{x_1}~~~\overbrace{ \hskip 2.2 in}^{x_2} \hskip 0.07 in \dots \hskip 0.07 in\overbrace{ \hskip 1.5 in }^{x_k}$
{\tiny
$$L_{\bf x}=\left[\begin{array}{cc|ccc|c|cc} \!\!\langle {\bf e}_1 |V| {\bf x} \rangle \!-\!v_{1,1} & -v_{1,1} & -v_{1,2} & -v_{1,2} & -v_{1,2} & \dots& -v_{1,k} & -v_{1,k} \\-v_{1,1} &  \!\!\!\!\!\!\langle {\bf e}_1 |V| {\bf x} \rangle \!-\!v_{1,1} & -v_{1,2} & -v_{1,2} & -v_{1,2} & \dots  & -v_{1,k} & -v_{1,k} \\ \hline -v_{2,1} & -v_{2,1} & \!\!\langle {\bf e}_2 |V| {\bf x} \rangle \!-\!v_{2,2} & -v_{2,2} & -v_{2,2} & \dots & -v_{2,k} & -v_{2,k} \\-v_{2,1} & -v_{2,1} & -v_{2,2} &  \!\!\!\!\!\!\langle {\bf e}_2 |V| {\bf x} \rangle \!-\!v_{2,2} & -v_{2,2} & \dots & -v_{2,k} & -v_{2,k} \\-v_{2,1} & -v_{2,1} & -v_{2,2} & -v_{2,2} &   \!\!\!\!\!\!\langle {\bf e}_2 |V| {\bf x} \rangle \!-\!v_{2,2} & \dots & -v_{2,k} & -v_{2,k} \\\hline \vdots & \vdots & \vdots & \vdots & \vdots & \ddots &  \vdots  & \vdots  \\\hline-v_{k,1} & -v_{k,1} & -v_{k,2} & -v_{k,2} & -v_{k,2} & \dots &  \!\!\langle {\bf e}_k |V| {\bf x} \rangle \!-\!v_{k,k} & -v_{k,k} \\-v_{k,1} & -v_{k,1} & -v_{k,2} & -v_{k,2} & -v_{k,2} & \dots & -v_{k,k} &  \!\!\!\!\!\!\langle {\bf e}_k |V| {\bf x} \rangle \!-\!v_{k,k}\end{array}\right]$$
}

\medskip
\noindent
Notice that for each $m=1,2,\hdots,k$ and each $j=s_{m-1}+1,\hdots,s_m$, vector $|{\bf e}_j \rangle - |{\bf e}_{j+1} \rangle \in \mathbb{R}^{\langle {\bf x} | \boldsymbol{1}\rangle}$ is an eigenvector of $L_{\bf x}$
corresponding to the eigenvalue $\langle {\bf e}_m |V| {\bf x} \rangle$. Hence, $\langle {\bf e}_m |V| {\bf x} \rangle$ is an eigenvalue of $L_{\bf x}$ of multiplicity $x_m-1$.

\bigskip
\noindent 
The weighted spanning tree enumerator $T_{\bf x}=T_{\bf x}(V)$ can be expressed via the celebrated Kirchhoff's Weighted Matrix-Tree Theorem \cite{Kirchhoff,KS20,Maxwell} as stated below.
\begin{thm}[Weighted Matrix-Tree Theorem]\label{thm:MatrixTree}
For any $1\leq i,j \leq \langle {\bf x} | {\bf 1} \rangle$, 
\be\label{eqn:MatrixTx}
T_{\bf x}=(-1)^{i+j} \det\big[L_{\bf x}\big]_{i,j},
\ee
where $\big[L_{\bf x}\big]_{i,j}$ denotes the $(i,j)$ minor of $L_{\bf x}$ obtained by removing the $i$-th row and $j$-th column in $L_{\bf x}$.
\end{thm}

\medskip
\noindent
Observe that for a simple graph $G$ with all edge weights $w_{i,j}=1$, the weighted spanning tree enumerator $\tau(G, w_{i,j})$
counts the number of spanning trees in $G$.
\begin{example}
In the 1-D case ($k=1$), $\,T_n = n^{n-2}$ is the number of spanning trees in a complete graph $K_n$, and equation \eqref{eqn:recTx}
turns into the following well known identity
\be\label{eqn:1dKnTn}
T_n = {1 \over 2(n-1) } \sum\limits_{m=1}^n {n \choose m} m (n-m) T_m T_{n-m}.
\ee
On the other hand, as discovered in \cite{McLeod62}, $S_n={n^{n-2} \over n!}$.
Thus, validating Lemma \ref{lem:solSx}.
\end{example}

\medskip
\noindent
\begin{example}
Let $V=|{\bf 1}\rangle \langle {\bf 1}|-I$.
In the context of the vector-multiplicative coalescent processes, this is the case when only the pairs of particles of different types are allowed to bond, 
each such pair bonding with rate $1/n$. 
Then, $T_{\bf x}=T_{\bf x}(V)$ is the number of spanning trees in a complete multipartite graph $\,K_{x_1,\hdots,x_k}$.
It was shown in \cite{Lewis1999} that the number of spanning trees in the complete multipartite graph equals
\be\label{eqn:KxTx}
T_{\bf x}  = n_{\bf x}^{k-2} \prod_{i=1}^k (n_{\bf x} - x_i)^{x_i-1}, \quad \text{ where }\quad n_{\bf x}=\langle {\bf x} | {\bf 1} \rangle.
\ee

\noindent
For instance, in the 2-D case ($k=2$),  $\,T_{x_1, x_2}= x_1^{x_2-1} x_2^{x_1-1}$ is the number of spanning trees in the complete bipartite graph $K_{x_1, x_2}$ 
with the partitions of sizes $x_1$ and $x_2$. Also, it was shown in \cite{KOY} that the solution $S_{x_1, x_2}$ of \eqref{SolnEqn} equals 
$S_{x_1, x_2} = {x_1^{x_2-1} x_2^{x_1-1} \over x_1! x_2!}$. Thus, Lemma \ref{lem:solSx} is validated for this case as well.
\end{example}

\bigskip
\noindent
Weighted Matrix-Tree Theorem (Thm.~\ref{thm:MatrixTree}) was  enhanced in S.~Klee and M.\,T. Stamps \cite{KS20} as follows.
\begin{lemma}[Weighted Matrix-Tree Lemma, \cite{KS20}]\label{lem:MatrixTreeKS}
For any given vectors $~{\bf a}, {\bf b} \in \mathbb{R}^{\langle {\bf x} | {\bf 1} \rangle}$ such that
$$\langle {\bf a} | {\bf 1} \rangle=\sum\limits_i a_i \not=0 \quad \text{ and }\quad \langle {\bf b} | {\bf 1} \rangle=\sum\limits_i b_i \not=0,$$ 
the weighted spanning tree enumerator equals
\be\label{eqn:abMatrixTx}
T_{\bf x}={\det\Big(L_{\bf x}+|{\bf a}\rangle\langle{\bf b}|\,\Big) \over \langle {\bf a} | {\bf 1} \rangle \langle {\bf b} | {\bf 1} \rangle}.
\ee
\end{lemma}

\bigskip
\noindent 
Lemma \ref{lem:MatrixTreeKS} will be used in our solution for $T_{\bf x}(V)$ in Thm.~\ref{thm:SolTx}. 
Notice that in $\,K_{\bf x}(V)$, the total weight of all edges connecting
vertices in the $i$-th partition with the vertices in the $j$-th partition equals $x_i x_j v_{i,j}$. Thinking of the $k$ partitions as $k$ vertices in 
the {\it partition graph} $K_k$, where vertex $i$ and vertex $j$, representing the corresponding partitions, are connected by an edge of weight $x_i x_j v_{i,j}$, 
the weighted enumerator for the spanning trees on the partition graph equals
\be\label{eqn:Tpartition}
\tau(K_k, x_i x_j v_{i,j})=T_{\bf 1}\big(D[x_i]VD[x_i]\big).
\ee
Notice that, by Weighted Matrix-Tree Theorem (Thm.~\ref{thm:MatrixTree}), 
\be\label{eqn:TpartitionL}
T_{\bf 1}\big(D[x_i]VD[x_i]\big)=(-1)^{i+j} \det\big[L(x_i x_j v_{i,j})\big]_{i,j} \quad \text{ for all }\quad 1 \leq i,j \leq k,
\ee
where the weighted Laplacian for the partition graph equals
$$L(x_i x_j v_{i,j})=D[x_i]\Big(D[\langle {\bf e}_i|V|{\bf x}\rangle]-VD[x_i]\Big).$$
\medskip
\noindent
Our next result reduces the computation of $T_{\bf x}(V)$ from $(\langle {\bf x} | {\bf 1} \rangle-1)$-dimensional determinants as in  \eqref{eqn:1dKnTn} and \eqref{eqn:abMatrixTx}, where $\langle {\bf x} | {\bf 1} \rangle$ gets arbitrarily large, to just computing the $(k-1)$-dimensional determinant in \eqref{eqn:TpartitionL}.
\begin{thm}[Solution for $T_{\bf x}(V)$]\label{thm:SolTx}
\be\label{eqn:SolTx}
T_{\bf x}(V)={\tau(K_k, x_i x_j v_{i,j}) \over {\bf x}^{\bf 1}}(V {\bf x})^{{\bf x}- {\bf 1}}.
\ee
\end{thm}
\begin{proof}
Recall that we let $s_0=0$, and $s_m=\sum\limits_{i=1}^m x_i$ for $m=1,\hdots,k$. 
Set $|{\bf a}\rangle=\!\sum\limits_{j=s_{k-1}+1}^{s_k} \!\!|{\bf e}_j\rangle$ in $\mathbb{R}^{\langle {\bf x} | {\bf 1} \rangle}$, i.e.,
$$\langle {\bf a}|=\big(0,\hdots,0,\underbrace{1,\hdots,1}_{x_k}\big),$$
and let
$$\langle {\bf b}|=\big(\underbrace{v_{k,1},\hdots,v_{k,1}}_{x_1}, ~\underbrace{v_{k,2},\hdots,v_{k,2}}_{x_2}, ~\hdots, ~\underbrace{v_{k,k},\hdots,v_{k,k}}_{x_k} \big).$$
Then, by Lemma \ref{lem:MatrixTreeKS}, 
\be\label{eqn:abMatrixPf}
T_{\bf x}={\det\Big(L_{\bf x}+|{\bf a}\rangle\langle{\bf b}|\,\Big) \over \langle {\bf a} | {\bf 1} \rangle \langle {\bf b} | {\bf 1} \rangle}
={\det\Big(L_{\bf x}+|{\bf a}\rangle\langle{\bf b}|\,\Big) \over x_k \langle {\bf e}_k |V| {\bf x} \rangle}.
\ee
On the other hand,
\be\label{eqn:Lxab}
L_{\bf x}+|{\bf a}\rangle\langle{\bf b}|=\left[\begin{array}{ccc|ccc} &  &  & * & * & * \\ & Q &  & * & * & * \\ &  &  & * & * & * \\\hline 0 & 0 & 0 &  \langle e_k |V| {\bf x} \rangle & 0 & 0 \\0 & 0 & 0 & 0 & \ddots & 0 \\0 & 0 & 0 & 0 & 0 &  \langle e_k |V| {\bf x} \rangle\end{array}\right],
\ee
where $s_{k-1}\times s_{k-1}$ matrix $Q=\Big[~L_{\bf x}~\Big]_{[1..s_{k-1}]\times[1..s_{k-1}]}$ is the restriction of the Laplacian matrix $L_{\bf x}$ to the first $s_{k-1}$ rows and columns. 

\medskip
\noindent
Equations \eqref{eqn:abMatrixPf} and \eqref{eqn:Lxab} yield
\be\label{eqn:TxQsk}
T_{\bf x}={\langle {\bf e}_k |V| {\bf x} \rangle^{x_k-1} \over x_k}\det(Q).
\ee

\medskip
\noindent
Recall the weighted Laplacian 
$$L(x_i x_j v_{i,j})=D[x_i]\Big(D[\langle {\bf e}_i|V|{\bf x}\rangle]-VD[x_i]\Big).$$
Therefore, by Thm.~\ref{thm:MatrixTree}, the weighted enumerator for the spanning trees on the partition graph equals
\be\label{eqn:tauQwave}
\tau(K_k, x_i x_j v_{i,j})=\det\big[L(x_i x_j v_{i,j})\big]_{k,k}=\left(\prod\limits_{j=1}^{k-1}x_j \right)\det(\widetilde{Q}),
\ee
where $\,\widetilde{Q}=\Big[D[\langle {\bf e}_i|V|{\bf x}\rangle]-VD[x_i]\Big]_{k,k}$ is the $(k,k)$ minor of $D[\langle {\bf e}_i|V|{\bf x}\rangle]-VD[x_i]$.

\medskip
\noindent
Next, we compare matrix
$$\widetilde{Q}=\left[\begin{array}{cccc}\langle {\bf e}_1 |V| {\bf x} \rangle -x_1 v_{1,1} & -x_2v_{1,2} & \dots & -x_{k-1} v_{1,k-1} \\-x_1 v_{2,1} & \langle {\bf e}_2 |V| {\bf x} \rangle -x_2 v_{2,2} & \dots & x_{k-1} v_{2,k-1} \\\vdots & \vdots & \ddots & \vdots \\-x_1 v_{k-1,1} & -x_2 v_{k-1,2} & \dots & \langle {\bf e}_{k-1} |V| {\bf x} \rangle -x_{k-1} v_{k-1,k-1}\end{array}\right]$$
to 

\medskip
$\hskip 0.27 in \overbrace{ \hskip 1.5 in }^{x_1}~~~\overbrace{ \hskip 1.45 in}^{x_2} \hskip 0.07 in \dots \hskip 0.07 in\overbrace{ \hskip 2.25 in }^{x_{k-1}}$
{\tiny
$$Q=\left[\begin{array}{cc|cc|c|cc} \!\!\langle {\bf e}_1 |V| {\bf x} \rangle \!-\!v_{1,1} & -v_{1,1} & -v_{1,2} & -v_{1,2} & \dots& -v_{1,k-1} & -v_{1,k-1} \\-v_{1,1} &  \!\!\!\!\!\!\langle {\bf e}_1 |V| {\bf x} \rangle \!-\!v_{1,1} & -v_{1,2} & -v_{1,2} & \dots  & -v_{1,k-1} & -v_{1,k-1} \\ \hline -v_{2,1} & -v_{2,1} & \!\!\langle {\bf e}_2 |V| {\bf x} \rangle \!-\!v_{2,2} & -v_{2,2} &  \dots & -v_{2,k-1} & -v_{2,k-1} \\-v_{2,1} & -v_{2,1} & -v_{2,2} &  \!\!\!\!\!\!\langle {\bf e}_2 |V| {\bf x} \rangle \!-\!v_{2,2} & \dots & -v_{2,k-1} & -v_{2,k-1} \\\hline \vdots & \vdots & \vdots & \vdots & \ddots &  \vdots  & \vdots  \\\hline-v_{k-1,1} & -v_{k-1,1} & -v_{k-1,2} & -v_{k-1,2}  & \dots &  \!\!\langle {\bf e}_{k-1} |V| {\bf x} \rangle \!-\!v_{k-1,k-1} & -v_{k-1,k-1} \\-v_{k-1,1} & -v_{k-1,1} & -v_{k-1,2} & -v_{k-1,2} & \dots & -v_{k-1,k-1} &  \!\!\!\!\!\!\langle {\bf e}_{k-1} |V| {\bf x} \rangle \!-\!v_{k-1,k-1}\end{array}\right]$$
}

\medskip
\noindent
First, we observe that if $|{\bf u}\rangle=(u_1,\hdots,u_{k-1})^T$ is a right eigenvector of $\widetilde{Q}$, then
$$\big(\underbrace{u_1,\hdots,u_1}_{x_1}, ~\underbrace{u_2,\hdots,u_2}_{x_2}, ~\hdots, ~\underbrace{u_{k-1},\hdots,u_{k-1}}_{x_{k-1}} \big)^T$$
is an eigenvector of $Q$ corresponding to the same eigenvalue.

\medskip
\noindent
Next, we find all $s_{k-1}\!-\!(k\!-\!1)=\sum\limits_{m=1}^{k-1}(x_m-1)$ remaining eigenvalues.
This is easy since for each $m=1,\hdots,k-1$ and each $j=s_{m-1}\!+\!1,\hdots,s_m$, 
vector $|{\bf e}_j \rangle \!-\! |{\bf e}_{j+1} \rangle$ in $\mathbb{R}^{s_{k-1}}$ is an eigenvector of $Q$
corresponding to the eigenvalue $\langle {\bf e}_m |V| {\bf x} \rangle$. 
Thus, $\langle {\bf e}_m |V| {\bf x} \rangle$ is an eigenvalue of $Q$ of multiplicity $x_m-1$.

\medskip
\noindent
Therefore,
\be\label{eqn:detQviaQtilde}
\det(Q)=\det(\widetilde{Q}) \prod\limits_{m=1}^{k-1} \langle {\bf e}_m |V| {\bf x} \rangle^{x_m-1}.
\ee
Together, equations \eqref{eqn:detQviaQtilde} and \eqref{eqn:tauQwave} imply
\be\label{eqn:detQ}
\det(Q)=T_{\bf 1}\big(D[x_i]VD[x_i]\big) \prod\limits_{j=1}^{k-1}x_j^{-1} \prod\limits_{m=1}^{k-1} \langle {\bf e}_m |V| {\bf x} \rangle^{x_m-1}.
\ee

\medskip
\noindent
Finally, substituting equation \eqref{eqn:detQ} into \eqref{eqn:TxQsk} yields
$$T_{\bf x}
=T_{\bf 1}\big(D[x_i]VD[x_i]\big) \prod\limits_{j=1}^k x_j^{-1} \prod\limits_{m=1}^k \langle {\bf e}_m |V| {\bf x} \rangle^{x_m-1}={T_{\bf 1}\big(D[x_i]VD[x_i]\big) \over {\bf x}^{\bf 1}}(V {\bf x})^{{\bf x}- {\bf 1}}.$$
Thus, by \eqref{eqn:Tpartition}, the proof is complete.
\end{proof}

\medskip
\begin{remark}
Observe that Theorem~\ref{thm:SolTx} reduces the need for calculating $T_{\bf x}(\cdot)$ for every ${\bf x}\in\mathbb{Z}_+^k$ 
to finding an expression for $T_{\bf 1}(\cdot)$, and substituting values of ${\bf x}$ into $T_{\bf 1}\big(D[x_i]VD[x_i]\big)$.
\end{remark}

\medskip
\noindent
\begin{example}
For a given vector ${\bf w} \in (0,\infty)^k$, let $$V = |{\bf w}\rangle \langle {\bf w}|-D[w_i^2].$$
Then, $K_{\bf x}(V)$ is a complete multipartite graph with weighted edges. 
Now,
$$D[x_i]VD[x_i]=|{\bf xw}\rangle \langle {\bf xw}|-D[x_i^2 w_i^2],$$
where $|{\bf xw}\rangle=\sum\limits_{i=1}^k x_i w_i |{\bf e}_i\rangle$ denotes the vector with coordinates $x_i w_i$.
Thus, the weighted Laplacian of the partition graph equals
$$L(x_i x_j v_{i,j})=\langle {\bf x}|{\bf w}\rangle D[x_i w_i]-|{\bf xw}\rangle \langle {\bf xw}|$$
and, by Lemma \ref{lem:MatrixTreeKS},
\be\label{eqn:exT1}
\tau(K_k, x_i x_j v_{i,j})={\det\Big(L(x_i x_j v_{i,j})+|{\bf xw}\rangle \langle {\bf xw}| \Big) \over \langle {\bf 1} |{\bf xw}\rangle^2}  
={\bf x}^{\bf 1} {\bf w}^{\bf 1} \langle {\bf x}|{\bf w}\rangle^{k-2}.
\ee
Hence, by Theorem~\ref{thm:SolTx}, substituting \eqref{eqn:exT1} into \eqref{eqn:SolTx} yields
\be\label{eqn:exTx}
T_{\bf x} = {\bf w}^{\bf 1} \langle {\bf w} | {\bf x} \rangle^{k-2} (V {\bf x})^{{\bf x}- {\bf 1}}.
\ee
Notice that letting ${\bf w}={\bf 1}$ in \eqref{eqn:exTx} yields \eqref{eqn:KxTx} as a special case.
\end{example}

\medskip
\noindent
Together, Proposition~\ref{prop:solODE}, Lemma~\ref{lem:solSx}, and Theorem \ref{thm:SolTx} yield the following 
general solution to the modified Smoluchowski equations  \eqref{eqn:Flory}.
\begin{cor}\label{cor:TxsolODE}
\be
\label{eqn:ODEsolnTxCor} 
\zeta_{\bf x}(t) = {1 \over {\bf x}!}\boldsymbol{\alpha}^{\bf x} T_{\bf x} e^{-\langle {\bf x} | V | \boldsymbol{\alpha} \rangle t} t^{\langle {\bf x} |\boldsymbol{1} \rangle -1}
\quad\text{ with }\quad T_{\bf x}={\tau(K_k, x_i x_j v_{i,j}) \over {\bf x}^{\bf 1}}(V {\bf x})^{{\bf x}- {\bf 1}} 
\ee
is the unique solution of MSE \eqref{eqn:Flory}.
\end{cor}
\noindent


\bigskip
\noindent
Finally, we would like to make the following general observation.
\begin{remark}
Notice that identity \eqref{eqn:1dKnTn} with $\,T_n = n^{n-2}$ is an application of Abel's binomial identity.  
In the 2-D case ($k=2$),  for $V=|{\bf 1}\rangle\langle {\bf 1}|-I$, the expression for weighted enumerator $\,T_{x_1, x_2}= x_1^{x_2-1} x_2^{x_1-1}\,$ 
can be obtained from a two dimensional generalization of Abel's identity in Huang and Liu \cite{HL}. See \cite{KOY}.
Thus, the multinomial equation \eqref{eqn:recTx} can be considered as a $k$-dimensional generalization of Abel's identity of the kind considered by
A.~Kelmans and A.~Postnikov in \cite{Postnikov}, J.~Pitman in \cite{Pitman2002}, and in related works.
Specifically, by Theorem~\ref{thm:SolTx}, for all nonnegative irreducible $V$, we have
\be\label{eqn:abeltau}
\tau(K_k, x_i x_j v_{i,j}) = \frac{1}{2 (\langle {\bf x} | \boldsymbol{1}\rangle - 1)} \sum_{{\bf y}, {\bf z} : {\bf y} + {\bf z} = {\bf x}} {{\bf x}! \over {\bf y}! {\bf z}!} \langle {\bf y} | V | {\bf z} \rangle \, \tau(K_k, y_i y_j v_{i,j}) \, \tau(K_k, z_i z_j v_{i,j}).
\ee

\end{remark}

\bigskip

\section{Gelation}

In the vector-multiplicative processes the total mass $\sum\limits_{\bf x} \zeta_{\bf x}(t) |{\bf x} \rangle$ is also a vector
with each coordinate being the corresponding component-vise total mass.
Consider the matrix of all second order moments $\,A(t)=\sum\limits_{\bf x} \zeta_{\bf x}(t) |{\bf x}\rangle \langle {\bf x}|$.
In this section, we will analyze $\sum\limits_{\bf x} \zeta_{\bf x}(t) |{\bf x} \rangle$ and $A(t)$ and establish gelation and find the gelation time.

\subsection{Divergence of second order moments}\label{sec:gel}
Let $\zeta_{\bf x}(t)$ be a solution to MSE \eqref{eqn:Flory}.
Then, \eqref{eqn:Flory} implies $\zeta_{\bf x}(t) \geq 0$ for all ${\bf x}$ and all $t\geq 0$.
For $n \in \mathbb{N}$, let  
$$|M_n(t)\rangle=\sum\limits_{{\bf x}: \langle {\bf 1} |{\bf x}\rangle \leq n} \zeta_{\bf x}(t) |{\bf x}\rangle
\quad\text{ and }\quad A_n(t)=\sum\limits_{{\bf x}: \langle {\bf 1} |{\bf x}\rangle \leq n} \zeta_{\bf x}(t) |{\bf x}\rangle \langle {\bf x}|$$
be the partial sums for the vector series $\sum\limits_{\bf x} \zeta_{\bf x}(t) |{\bf x} \rangle$ and matrix series $\,A(t)=\sum\limits_{\bf x} \zeta_{\bf x}(t) |{\bf x}\rangle \langle {\bf x}|$.
Equation \eqref{eqn:Flory} yields the following inequality
\begin{align}\label{eqn:ineqMnt}
{d \over dt}|M_n(t)\rangle &=- \sum_{{\bf x}: \langle {\bf 1} |{\bf x}\rangle \leq n} \zeta_{\bf x}(t) |{\bf x} \rangle \langle {\bf x} | V | \boldsymbol{\alpha} \rangle  + \frac{1}{2}\sum_{{\bf y}, {\bf z} \,:  \langle {\bf 1} |{\bf y} + {\bf z}\rangle \leq n} \langle {\bf y} | V | {\bf z} \rangle \zeta_{{\bf y}} \zeta_{{\bf z}} | {\bf y} +{\bf z} \rangle \nonumber \\
&=- A_n(t) V | \boldsymbol{\alpha} \rangle  + \sum_{{\bf y}, {\bf z} \,:  \langle {\bf 1} |{\bf y} + {\bf z}\rangle \leq n}  \zeta_{{\bf y}} | {\bf y} \rangle\langle {\bf y} | V | {\bf z} \rangle \zeta_{{\bf z}} \nonumber \\
&< - A_n(t) V | \boldsymbol{\alpha} \rangle  + \sum_{\substack{{\bf y}:  \langle {\bf 1} |{\bf y} \rangle \leq n \\ {\bf z} : \langle {\bf 1} |{\bf z}\rangle \leq n}}  \zeta_{{\bf y}} | {\bf y} \rangle\langle {\bf y} | V | {\bf z} \rangle \zeta_{{\bf z}} \nonumber \\
&=- A_n(t) V | \boldsymbol{\alpha} \rangle  +A_n(t) V | M_n(t)\rangle \nonumber \\
&=- A_n(t)V \big | \boldsymbol{\alpha} - M_n(t) \big\rangle,
\end{align}
where $\,|M_n(0)\rangle=| \boldsymbol{\alpha} \big\rangle$.
Inequality \eqref{eqn:ineqMnt} implies 
$$|M_n(t)\rangle \leq | \boldsymbol{\alpha} \big\rangle \qquad \forall t\geq 0.$$
Therefore,
$$\sum\limits_{{\bf x}: \langle {\bf 1} |{\bf x}\rangle \leq n} \zeta_{\bf x}(t)  \leq  \langle {\bf 1}|M_n(t)\rangle \leq  \langle {\bf 1} | \boldsymbol{\alpha} \big\rangle$$
and series
$$\sum\limits_{\bf x}\zeta_{\bf x}(t) \quad \text{ is convergent for all }~\boldsymbol{\alpha} \in (0,\infty)^k~\text{ and all }\,t \geq 0.$$

\medskip
\noindent
Let $S({\bf z})=\sum\limits_{\bf x}S_{\bf x}{\bf z}^{\bf x}$ be the generating function of $S_{\bf x}$ defined as a $k$-dimensional power series.
Notice that by \eqref{eqn:testsoln} we have 
\be\label{eqn:zetaSxw}
\zeta_{\bf x}(t) = \boldsymbol{\alpha}^{\bf x} S_{\bf x} e^{-\langle {\bf x} | V | \boldsymbol{\alpha} \rangle t} t^{\langle {\bf x} | \boldsymbol{1} \rangle -1}
={1 \over t}S_{\bf x}\left(\sum\limits_{j=1}^k \alpha_j t  e^{-\langle {\bf e}_j  | V | {\boldsymbol \alpha} t \rangle} |{\bf e}_j\rangle\right)^{\bf x}={1 \over t}S_{\bf x}{\bf w}^{\bf x},
\ee
where $\,|{\bf w}\rangle=\sum\limits_{j=1}^k \alpha_j t  e^{-\langle {\bf e}_j  | V | {\boldsymbol \alpha} t \rangle} |{\bf e}_j\rangle$.

\medskip
\noindent
Therefore, since 
$$\sum\limits_{\bf x}\zeta_{\bf x}(t)={1 \over t}\sum\limits_{\bf x} S_{\bf x}{\bf w}^{\bf x}$$ 
converges for all choices of $\,\boldsymbol{\alpha} \in (0,\infty)^k\,$ and all $t>0$,
the series $\,\sum\limits_{\bf x}S_{\bf x}{\bf w}^{\bf x}\,$
converges for all ${\bf w}$ in the domain $\mathcal{D}$ as in \eqref{eqn:DR}.
Hence, by \eqref{eqn:DR0}, open set
$$\mathcal{D}_0=\left\{{\bf u} \in (0,\infty)^k ~~:~~\exists {\bf z} \in  R_0  \quad\text{ such that }\quad
|{\bf u}\rangle =\sum\limits_{j=1}^k z_j \,  e^{-\langle {\bf e}_j  | V | {\bf z} \rangle}|{\bf e}_j\rangle  \right\}.$$
is a subset of the domain (interior region) of convergence of $S({\bf z})=\sum\limits_{\bf x}S_{\bf x}{\bf z}^{\bf x}$.

\begin{lemma}\label{lem:2momentGelation}
Solutions of \eqref{eqn:SE}  and \eqref{eqn:Flory} coincide for $\,0 \leq t \leq {1 \over \rho ( V D[\alpha_i] )}$.
Moreover,
\be\label{eqn:conservmass}
\sum_{\bf x} \zeta_{\bf x}(t) |{\bf x} \rangle = |\boldsymbol{\alpha} \rangle \quad \text{ for all }~0 \leq t < {1 \over \rho (V D[\alpha_i] )}.
\ee
\end{lemma}
\begin{proof}
Observe that for all multinomials $p({\bf x})=p(x_1,\hdots,x_k)$, the series $\sum\limits_{\bf x} p({\bf x}) S_{\bf x}{\bf z}^{\bf x}$
is convergent in the domain of convergence of $S({\bf z})=\sum\limits_{\bf x}S_{\bf x}{\bf z}^{\bf x}$.
Thus, by \eqref{eqn:zetaSxw}, series 
$$\sum\limits_{\bf x}p({\bf x}) \zeta_{\bf x}(t)={1 \over t}\sum\limits_{\bf x} p({\bf x}) S_{\bf x}{\bf w}^{\bf x} \quad\text{ with }\quad |{\bf w}\rangle=\sum\limits_{j=1}^k \alpha_j t  e^{-\langle {\bf e}_j  | V | {\boldsymbol \alpha} t \rangle} |{\bf e}_j\rangle$$
converges whenever ${\bf w} \in \mathcal{D}_0$. This happens when ${\boldsymbol \alpha} t \in R_0$, i.e., when
$$t<{1 \over \rho (V D[\alpha_i] )}.$$
Since $p({\bf x})$ can be taken to be quadratic, the matrix of all second order moments $\,A(t)=\sum\limits_{\bf x} \zeta_{\bf x}(t) |{\bf x}\rangle \langle {\bf x}|$
is finite when $t<{1 \over \rho (V D[\alpha_i] )}$. 
Hence, equation \eqref{eqn:consmassSE} implies that $\,\sum_{\bf x} \zeta_{\bf x}(t) |{\bf x} \rangle = |\boldsymbol{\alpha} \rangle$ for all $\,0 \leq t <{1 \over \rho ( V D[\alpha_i] )}$.
\end{proof}

\bigskip
\noindent 
Additionally, equation \eqref{eqn:SE} implies
\begin{align*}
{d \over dt}A(t) &= \frac{1}{2}\sum_{{\bf y}, {\bf z}} \langle {\bf y} | V | {\bf z} \rangle \zeta_{{\bf y}}(t) \zeta_{{\bf z}}(t) \Big( |{\bf y+z}\rangle \langle {\bf y+z}| \Big) 
- \sum_{{\bf x,y}} \langle {\bf x} | V | {\bf y} \rangle \zeta_{\bf x}(t) \zeta_{{\bf y}}(t) \Big( |{\bf x}\rangle \langle {\bf x}| \Big) \\
&= \frac{1}{2}\sum_{{\bf y}, {\bf z}} \langle {\bf y} | V | {\bf z} \rangle \zeta_{{\bf y}}(t) \zeta_{{\bf z}}(t) \Big( |{\bf y}\rangle \langle {\bf z}|+|{\bf z}\rangle \langle {\bf y}| \Big) 
~= \sum_{{\bf y}, {\bf z}} \langle {\bf y} | V | {\bf z} \rangle \zeta_{{\bf y}}(t) \zeta_{{\bf z}}(t) \Big( |{\bf y}\rangle \langle {\bf z}| \Big)\\
&= \left(\sum_{\bf y} \zeta_{{\bf y}}(t) |{\bf y}\rangle \langle {\bf y}|\right) V \left(\sum_{\bf z} \zeta_{{\bf z}}(t) |{\bf z}\rangle \langle {\bf z}|\right) 
~=A(t)\,VA(t),
\end{align*}
and therefore,
\begin{equation}\label{eqn:odeAVt}
{d \over dt}\big(VA(t)\big)=\big(VA(t)\big)^2
\end{equation}
\\
\noindent
with the initial conditions $A(0)=D[\alpha_i]$.
Note that we used finiteness of some third order moments of $\zeta_{\bf x}(t)$. However, they should be finite for all positive $t<{1 \over \rho (V D[\alpha_i] )}$
as $p({\bf x})$ can be take to be a third degree multinomial.

\medskip
\noindent
Naturally, equation \eqref{eqn:odeAVt} has the following solution
\be\label{eqn:solAt}
A(t)=D[\alpha_i]\Big(I-tVD[\alpha_i]\Big)^{-1}
\ee
implying the explosive behavior of the second moments matrix $A(t)$ as $\,t \uparrow {1 \over \rho (V D[\alpha_i] )}$.

\bigskip

\subsection{Gelation via mass conservation and mass dissipation} 
As a consequence of multidimensional Lambert-Euler inversion,
$\sum\limits_{\bf x} \zeta_{\bf x}(t) |{\bf x} \rangle$ can be expressed via the minimal solution ${\bf y}(t)$ of \eqref{eqn:GenMin}.

\medskip
\begin{lemma}\label{len:massloss} 
Let $\zeta_{\bf x}(t)$ denote the solution to the modified Smoluchowski equation \eqref{eqn:Flory} and ${\bf y}(t)=\Lambda_V \circ \Psi_V({\boldsymbol \alpha} t)$ be the minimal solution of \eqref{eqn:GenMin}.  Then,
$$\sum_{\bf x}  \zeta_{\bf x}(t)|{\bf x}\rangle = {1 \over t}|{\bf y}(t)\rangle, \quad\text{ or equivalently, } \qquad \sum_{\bf x} x_i \zeta_{\bf x}(t) = {y_i(t) \over t} \qquad (i=1,\hdots,k).$$
\end{lemma}
\bp
Consider the generating function  $\,S({\bf z})=\sum\limits_{\bf x}S_{\bf x}{\bf z}^{\bf x}\,$ of $\,S_{\bf x}\,$ for $\,{\bf z} \in \mathbb{R}^k$ 
in the domain of convergence.
Recall that, by \eqref{eqn:zetaSxw}, we have the following representation
$$\zeta_{\bf x}(t) ={1 \over t}S_{\bf x}{\bf w}^{\bf x}, \quad \text{ where } \quad |{\bf w}\rangle=\sum\limits_{j=1}^k \alpha_j t  e^{-\langle {\bf e}_j  | V | {\boldsymbol \alpha} t \rangle} |{\bf e}_j\rangle.$$
Therefore, 
$$\zeta_{\bf x}(t)\,x_i={1 \over t} x_i S_{\bf x}{\bf w}^{\bf x}={1 \over t} w_i {\partial \over \partial w_i} S_{\bf x}{\bf w}^{\bf x},$$ 
and
$$\sum_{\bf x}  \zeta_{\bf x}(t)|{\bf x}\rangle = {1 \over t} D[w_i] \big|\nabla S({\bf w}) \big\rangle \quad \text{ with the gradient of }\,S({\bf z})\,\text{ taken at }~~{\bf w}.$$
Thus, by Lemma \ref{lem:2momentGelation}, we have
$$t|{\boldsymbol \alpha}\rangle = D[w_i] \big| \nabla S({\bf w}) \big\rangle, \quad \text{ where }\quad  |{\bf w}\rangle=\sum\limits_{j=1}^k \alpha_j t  e^{-\langle {\bf e}_j  | V | {\boldsymbol \alpha} t \rangle} |{\bf e}_j\rangle,$$
for any choice of $\,{\boldsymbol \alpha}>{\bf 0}\,$ and $\,0<t < {1 \over \rho ( V D[\alpha_i] )}$.

\medskip
\noindent
Hence, for all ${\bf y} \in R_0$, we have
\be\label{eqn:gradSwy}
|{\bf y}\rangle =  D[w_i] \big| \nabla S({\bf w}) \big\rangle, \quad \text{ where }\quad  |{\bf w}\rangle=\sum\limits_{j=1}^k y_j  e^{-\langle {\bf e}_j  | V | {\bf y} \rangle} |{\bf e}_j\rangle.
\ee

\medskip
\noindent
Next, for a given ${\boldsymbol \alpha}>{\bf 0}\,$ and $\,t \not= {1 \over \rho ( V D[\alpha_i] )}$, let ${\bf y}(t)=\Lambda_V \circ \Psi_V({\boldsymbol \alpha} t)$ be the minimal solution of \eqref{eqn:GenMin}. 
Then, we have
$$\sum_{\bf x}  \zeta_{\bf x}(t)|{\bf x}\rangle = {1 \over t} D[w_i] | \nabla S({\bf w})\rangle,$$
where, by \eqref{eqn:GenMin},
$$|{\bf w}\rangle=\sum\limits_{j=1}^k \alpha_j t  e^{-\langle {\bf e}_j  | V | {\boldsymbol \alpha} t \rangle} |{\bf e}_j\rangle=\sum\limits_{j=1}^k y_j  e^{-\langle {\bf e}_j  | V | {\bf y}(t) \rangle} |{\bf e}_j\rangle$$
and as ${\bf y}(t) \in R_0$, \eqref{eqn:gradSwy} yields
$$D[w_i] \big| \nabla S({\bf w})\big\rangle=|{\bf y}(t)\rangle.$$
Therefore,
$$\sum_{\bf x}  \zeta_{\bf x}(t)|{\bf x}\rangle ={1 \over t} D[w_i] \big| \nabla S({\bf w})\big\rangle= {1 \over t}  |{\bf y}(t)\rangle$$
affirming the statement of the lemma for $\,t \not= {1 \over \rho ( V D[\alpha_i] )}$.

\medskip
\noindent
Now, equation \eqref{eqn:ineqMnt} implies that the partial sums $\,|M_n(t)\rangle\,$ of $\,\sum_{\bf x}  \zeta_{\bf x}(t)|{\bf x}\rangle\,$ are decreasing, 
and therefore, $\sum_{\bf x}  \zeta_{\bf x}(t)|{\bf x}\rangle$ itself is coordinate-wise nonincreasing. Thus, $\,\sum_{\bf x}  \zeta_{\bf x}(t)|{\bf x}\rangle \leq |{\boldsymbol \alpha}\rangle$, and by continuity of ${\bf y}(t)=\Lambda_V \circ \Psi_V({\boldsymbol \alpha} t)$, we have
$$|{\boldsymbol \alpha}\rangle =\lim\limits _{~~t \downarrow {1 \over \rho ( V D[\alpha_i] )}}{1 \over t}  |{\bf y}(t)\rangle 
=\lim\limits _{~~t \downarrow {1 \over \rho ( V D[\alpha_i] )}}\sum_{\bf x}  \zeta_{\bf x}(t)|{\bf x}\rangle 
~\leq~ \sum_{\bf x}  \zeta_{\bf x}\left({1 \over \rho ( V D[\alpha_i] )}\right)\big|{\bf x}\big\rangle ~\leq~ |{\boldsymbol \alpha}\rangle.$$
This yields $\sum_{\bf x}  \zeta_{\bf x}(t)|{\bf x}\rangle =|{\boldsymbol \alpha}\rangle= {1 \over t}  |{\bf y}(t)\rangle$ for $\,t = {1 \over \rho ( V D[\alpha_i] )}$.
\ep

\bigskip
\noindent
Gelation time $T_{gel}$ is defined as the time until which the mass $\sum\limits_{\bf x} \zeta_{\bf x}(t) |{\bf x} \rangle$ is conserved, and after which, 
the mass begins to dissipate coordinate-wise.
\begin{defn}\label{def:gel}
The gelation time is the infimum 
$$T_{gel}=\inf\Big\{t>0\,:\,\sum_{\bf x} \zeta_{\bf x}(t) |{\bf x} \rangle  ~<~  |\boldsymbol{\alpha} \rangle \Big\}.$$
\end{defn}

\medskip
\noindent
Together, Corollary \ref{cor:ytat} and Lemma \ref{len:massloss} imply the following.
\begin{cor}\label{cor:gel}
Let $\zeta_{\bf x}(t)$ be the solution to the modified Smoluchowski equation \eqref{eqn:Flory}.  Then, gelation time equals 
$$T_{gel}={1 \over \rho ( V D[\alpha_i] )}.$$
\end{cor}

\bigskip
\noindent
Recall another critical time introduced in Subsection \ref{sec:Smol}
$$t_c =\inf\Big\{t>0~:~A(t)=\sum\limits_{\bf x} \zeta_{\bf x}(t) |{\bf x}\rangle \langle {\bf x}| \text{ diverges } \Big\}.$$
From Lemma \ref{len:massloss} and equation \eqref{eqn:consmassSE} we deduce the value of $t_c$.
\begin{cor}\label{cor:tcp}
For $\zeta_{\bf x}(t)$ solving the modified Smoluchowski equation \eqref{eqn:Flory},
$$t_c ={1 \over \rho ( V D[\alpha_i] )}.$$
\end{cor}
\begin{proof}
Equation \eqref{eqn:consmassSE} yields $\sum_{\bf x} \zeta_{\bf x}(t) |{\bf x} \rangle= |\boldsymbol{\alpha} \rangle$ for all $t \in [0,t_c)$.
Thus, Lemma \ref{len:massloss} implies  $\, t_c \leq {1 \over \rho ( V D[\alpha_i] )} $.
On the other hand, the second moment matrix series $A(t)=\sum\limits_{\bf x} \zeta_{\bf x}(t) |{\bf x}\rangle \langle {\bf x}| $ is finite
for  $\, t <{1 \over \rho ( V D[\alpha_i] )}$ (recall \eqref{eqn:solAt} and its derivation), yielding $\, t_c \geq {1 \over \rho ( V D[\alpha_i] )}$.
\end{proof}

\bigskip
\noindent
Lemmas \ref{lem:massdissipation} and \ref{len:massloss} yield another important corollary,
stating that the total mass $\,\sum\limits_{\bf x} \zeta_{\bf x}(t) |{\bf x} \rangle\,$ will eventually dissipate to nothing,
corresponding to a fact that all smaller clusters will be eventually absorbed by a giant component.
\begin{cor}\label{cor:massto0}
Let $\zeta_{\bf x}(t)$ be the solution to the modified Smoluchowski equation \eqref{eqn:Flory}.  Then, 
$$\lim\limits_{t \to \infty}\sum\limits_{\bf x} \zeta_{\bf x}(t) |{\bf x} \rangle={\boldsymbol 0}.$$
\end{cor}

\bigskip

\section{Application in Minimal Spanning Trees}\label{sec:mst} 

For a nonnegative irreducible symmetric matrix $V \in \mathbb{R}^{k \times k}$ and a vector ${\boldsymbol \alpha}\in(0,\infty)^k$, 
let ${\boldsymbol \alpha}[n]$ be as in \eqref{eqn:aplhan}. Consider the graph $K_{{\boldsymbol \alpha}[n]}$ equipped with random edge lengths $\ell_e$ 
as defined in Sect.~\ref{intro:mst}. Recall that the length of a tree is the sum of the lengths of the tree's edges, and 
let the random variable $L_n$ denote the length of the minimal spanning tree of $K_{{\boldsymbol \alpha}[n]}$. 
We are interested in the asymptotic mean lengths of the minimal spanning tree of $K_{{\boldsymbol \alpha}[n]}$ as $n \to \infty$. 
The following theorem follows immediately from Sect. 4.3 in \cite{KOY}.
\begin{thm}\label{thm:mainKan}
\begin{equation}\label{eqn:mst}
\lim_{n \goto \infty} \mathbb{E}[L_n] = \sum_{\bf x} \int\limits_0^\infty \zeta_{\bf x}(t) \, dt,
\end{equation}
where $\zeta_{\bf x}(t)$ is the solution of the modified Smoluchowski coagulation system  \eqref{eqn:Flory}.
\end{thm}
\begin{proof}
Follows immediately from the proof in Section 4.3 of \cite{KOY} by replacing $V=\left[\begin{array}{cc}0 & 1 \\1 & 0\end{array}\right]$
with any other nonnegative irreducible symmetric $V \in \mathbb{R}^{k \times k}$.
\end{proof}

\medskip
\noindent
Applying Corollary \ref{cor:TxsolODE} and Theorem \ref{thm:mainKan} together results in the following closed form expression 
for the limit $\lim\limits_{n \goto \infty} \mathbb{E}[L_n]$.
\begin{cor}\label{cor:LnKxV}
\be\label{eqn:LnKxV}
\lim_{n \goto \infty} \mathbb{E}[L_n] = \sum_{\bf x} {(\langle {\bf x} |\boldsymbol{1} \rangle -1)! \over {\bf x}!}\boldsymbol{\alpha}^{\bf x} T_{\bf x} \, \langle {\bf x} | V | \boldsymbol{\alpha} \rangle^{-\langle {\bf x} | {\bf 1} \rangle}  
\ee
$$\text{ with }\quad T_{\bf x}={\tau(K_k, x_i x_j v_{i,j}) \over {\bf x}^{\bf 1}}(V {\bf x})^{{\bf x}- {\bf 1}}.$$
\medskip
\end{cor}
\begin{proof}
Substituting 
$\,\zeta_{\bf x}(t) = {1 \over {\bf x}!}\boldsymbol{\alpha}^{\bf x} T_{\bf x} e^{-\langle {\bf x} | V | \boldsymbol{\alpha} \rangle t} t^{\langle {\bf x} |\boldsymbol{1} \rangle -1}\,$ 
from equation \eqref{eqn:ODEsolnTxCor} into equation \eqref{eqn:mst} yields
$$\lim_{n \goto \infty} \mathbb{E}[L_n] = \sum_{\bf x}  {1 \over {\bf x}!}\boldsymbol{\alpha}^{\bf x} T_{\bf x} \int\limits_0^\infty e^{-\langle {\bf x} | V | \boldsymbol{\alpha} \rangle t} t^{\langle {\bf x} |\boldsymbol{1} \rangle -1} \, dt
= \sum_{\bf x} {1 \over {\bf x}!}\boldsymbol{\alpha}^{\bf x} T_{\bf x} \, \langle {\bf x} | V | \boldsymbol{\alpha} \rangle^{-\langle {\bf x} | {\bf 1} \rangle} \Gamma\big(\langle {\bf x} |\boldsymbol{1} \rangle\big).$$
\end{proof}

\medskip
\noindent
Now, we will use the following example to validate the general formula \eqref{eqn:LnKxV} in Corollary \ref{cor:LnKxV}.
\begin{example}\label{ex:zeta}
Let $V=|{\bf 1}\rangle \langle {\bf 1}|-I$, i.e., $K_{{\boldsymbol \alpha}[n]}$ is a complete multipartite graph with edge lengths $\ell_e$ 
uniformly distributed on $(0,1)$.
Substituting \eqref{eqn:KxTx} into \eqref{eqn:LnKxV} yields
\be\label{eqn:exKxLnComplete}
\lim_{n \goto \infty} \mathbb{E}[L_n] = \sum_{\bf x} {(n_{\bf x} -1)! \over {\bf x}!}\boldsymbol{\alpha}^{\bf x} n_{\bf x}^{k-2} 
\, \big(n_{\bf x}\langle {\bf 1} | \boldsymbol{\alpha} \rangle -  \langle {\bf x} | \boldsymbol{\alpha} \rangle\big)^{-n_{\bf x}} \prod_{i=1}^k (n_{\bf x} - x_i)^{x_i-1} , 
\ee
where $\, n_{\bf x}=\langle {\bf x} | {\bf 1} \rangle$.
In the equipartitioned case, when $\boldsymbol{\alpha}={\bf 1}$, equation \eqref{eqn:exKxLnComplete} simplifies to
\begin{align}\label{eqn:exEquiPart}
\lim_{n \goto \infty} \mathbb{E}[L_n] &= \sum_{\bf x} {n_{\bf x}! \over {\bf x}!} \,n_{\bf x}^{k-n_{\bf x}-3}\, (k-1)^{-n_{\bf x}} \prod_{i=1}^k (n_{\bf x} - x_i)^{x_i-1} \nonumber \\ 
&= \sum\limits_{n=1}^\infty n^{k-n-3}\, (k-1)^{-n} \!\!\sum\limits_{{\bf x}:\langle {\bf x} | {\bf 1}\rangle=n} {n! \over {\bf x}!}  \prod_{i=1}^k (n- x_i)^{x_i-1}
\end{align}
Now, by the Abel's type multinomial identity from Abramson \cite{Abramson1969}, we have 
\be\label{eqn:newAbelnx}
\sum\limits_{{\bf x}:\langle {\bf x} | {\bf 1}\rangle=n} {n! \over {\bf x}!}  \prod_{i=1}^k (n- x_i)^{x_i-1}=k (k-1)^{n-1} n^{n-k}.
\ee
Substituting \eqref{eqn:newAbelnx} into \eqref{eqn:exEquiPart}, we get
\be\label{eqn:kzeta3}
\lim_{n \goto \infty} \mathbb{E}[L_n] ={k \over k-1} \sum\limits_{n=1}^\infty n^{-3}={k \over k-1} \zeta(3),
\ee
where $\zeta(n)$ is Riemann zeta function.
Equation \eqref{eqn:kzeta3} matches the general expression for the asymptotic limit 
$\lim_{n \goto \infty} \mathbb{E}[L_n]$ for regular graphs with i.i.d. uniformly distributed edge lengths as derived in Beveridge et al \cite{BFMcD},
thus, validating the general formula \eqref{eqn:LnKxV} in Corollary \ref{cor:LnKxV}.
\end{example}

\bigskip

\section*{Acknowledgements}
This research was supported in part by NSF award DMS-1412557.

\bibliographystyle{amsplain}

\end{document}